%% file: fin-alg.tex
\documentclass[a4paper,twoside]{article}

\input{prelude}

\input{theorems}

\author{
        Christine Tasson${^1}$\thanks{%
        This work has been partially funded by the French ANR projet
        blanc ``Curry Howard pour la Concurrence'' CHOCO ANR-07-BLAN-0324.}
				\ and Lionel Vaux${^2}$\footnotemark[1]\\
				\normalsize
        $^1$ Preuves, Programmes et Systèmes,
        CNRS UMR 7126, Paris, France.
        \\
				\normalsize
        $^2$ Institut de Mathématiques de Luminy,
        CNRS UMR 6206, Marseille, France.
}

\newcommand{\thetitle}{%
Transport of finiteness structures and applications
}
\title{\thetitle}

\date{December 31, 2010}

\begin{document}

\maketitle

\begin{abstract}
  We describe a general construction of finiteness spaces which subsumes
  the interpretations of all positive connectors of linear logic. We then
  show how to apply this construction to prove the existence of least
  fixpoints for particular functors in the category of finiteness spaces:
  these include the functors involved in a relational interpretation of
  lazy recursive algebraic datatypes along the lines of the coherence
  semantics of system T.
\end{abstract}

\section{Introduction}

Finiteness spaces were introduced by \citet{ehrhard:fs}, refining the
purely relational model of linear logic. A finiteness space is a set equipped
with a finiteness structure, \ie{} a particular set of subsets which are said
to be finitary; and the model is such that the relational denotation of a proof
in linear logic is always a finitary subset of its conclusion. 
Applied to this finitary relational  model of linear logic,  the usual co-Kleisli
construction provides a cartesian closed category, hence a model of
the simply typed $\lambda$-calculus (see, \eg, \citet{bierman:catmodel}).
The crucial property of finiteness spaces is
that the intersection of two finitary subsets of dual types is always finite.
This feature allows to reformulate the quantitative semantics of
\citet{girard:quantitative} in a standard algebraic setting, where morphisms
interpreting typed $\lambda$-terms are analytic functions between the
topological vector spaces generated by vectors with finitary supports. This
provided the semantic foundations of the differential
$\lambda$-calculus of \citet{er:tdlc} and motivated the general study of a
differential extension of linear logic \citep[\emph{etc.}]{%
er:diffnets,er:bkt,el:pi,tranquilli:idn,vaux:poldiff,tasson:totality,pt:inverse-taylor}.

The fact that finiteness spaces form a model of linear logic can be understood
as a property of the relational interpretation: as we have already mentioned,
the relational semantics of a proof is always finitary. The present paper
studies the connexion between the category $\Rel$ of sets and relations and
the category $\Fin$ of finiteness spaces and finitary relations, while
maintaining a similar standpoint: we investigate whether and how some of the
most distinctive features of $\Rel$ can be given counterparts in $\Fin$. 

Our primary contribution is a very general construction of finiteness spaces:
given a relation from a set $A$ to a finiteness space such that the
relational image of every element is finitary, we can form a new finiteness
space on $A$ whose finitary subsets are exactly those with finitary image. We refer to
this result as the \definitive{transport lemma}. Although simple in its
formulation, the transport lemma subsumes many constructions in finiteness
spaces and in particular those interpreting the positive connectives of linear
logic (multiplicative ``$\tensor$'', additive ``$\oplus$'' and exponential
``$\oc$'') whose action on sets is given by the corresponding relational
interpretations. We moreover provide sufficient conditions for a functor in
$\Rel$ to give rise to a functor in $\Fin$ \emph{via} the transport lemma:
again, this generalizes the functoriality of the positive connectives of linear
logic.

The category $\Rel$, endowed with inclusion on sets and relations, is enriched
on complete partial orders (cpo). This structure was studied in a more general
2-categorical setting \citep{cks:spans-and-relations,ckw:change-of-base} and the
properties of monotonic functors allowed for an abstract description of
datatypes \citep{backhouse-et-al:polrel,hdm:contcat,bh:gpd}. In
such a setting, it is standard to define recursive datatypes, such as lists or
trees, as the least fixpoints of particular Scott-continuous functors
\citep{ps:recursive}. This
prompted us to consider two orders on finiteness spaces derived from set
inclusion: the most restrictive one, \definitive{finiteness extension}, was
used by \citet[unpublished preliminary version]{ehrhard:fs} to provide an
interpretation of second order linear logic, while the largest one,
\definitive{finiteness inclusion}, is a cpo on finiteness
spaces. We study various notions of continuity for functors in finiteness
spaces, and relate them with the existence of fixpoints. A striking feature of
this development is that we are led to consider the properties of functors
w.r.t. both orders simultaneously: continuity for finiteness inclusion, and
monotonicity for finiteness extension. We prove in particular that every
functor obtained by applying the transport lemma to a continuous relational
functor satisfies these properties, and admits a least fixpoint for finiteness
inclusion.

The remaining of the paper is dedicated to the application of these results to
the relational semantics of functional programming with recursive datatypes.
Indeed, the co-Kleisli construction applied to the relational model of linear
logic gives rise to the cartesian closed category $\RelCC$. The fact that the
already mentioned co-Keisli $\FinCC$ of $\Fin$ provides a model of the
$\lambda$-calculus can again be understood as a property of the interpretation
in $\RelCC$: the relational semantics of a \emph{simply typed} $\lambda$-term
is always finitary.  It is however worth noticing that, whereas the relational
model can accomodate untyped $\lambda$-calculi \citep{carvalho:exec,bem:enough},
finiteness spaces are essentially a model of termination. The whole point of
the finiteness construction is to reject infinite computations, ensuring that
the intermediate sets involved in the relational interpretation of a cut are
all finite. In particular, the relational semantics of fixpoint combinators is
finitary only on empty types: general recursion is ruled out from this
framework.  This is to be related with the fact that finitary relations are
\emph{not} closed under arbitrary unions: in contrast with the cpo defined on
objects by set inclusion,
the category $\FinCC$ (and thus $\Fin$) is not enriched on complete partial
orders. 

Despite this restrictive design, \citet[Section 3]{ehrhard:fs} was able to
define a finitary interpretation of tail-recursive iteration: this
indicates that the finiteness semantics can accomodate a form of typed
recursion. This interpretation, however, is not completely satisfactory: 
tail recursive iteration is essentially linear, thus it does not provide a type of
natural numbers \citep{thibault:prerecursive,ls:hocl} in the associated model of
the $\lambda$-calculus. This is essentially due to the fact that the
interpretation of natural numbers is \emph{flat} (in the sense of domains). In
fact, a similar effect was already noted by Girard in the design of his coherence
semantics of system $T$ \citep{girard:prot}: his solution was to propose a
\emph{lazy} interpretation of natural numbers, where laziness refers to the
possibility of pattern matching on non normal terms. The second author remarked
that the same solution could be adapted in the relational
model and provided a type of natural numbers with finitary recursor, hence a
model of system $T$~\citep{vaux:relT}.

Our previous developments allow us to generalize this construction: after
introducing a finitary relational interpretation of sum types, we consider the
fixpoints of particular functors and show that they provide a relational semantics
of the typed $\lambda$-calculus with lazy recursive algebraic datatypes by
exhibiting their constructors and destructors. Adapting the techniques already
employed by the  second author in the case of  system $T$, we moreover
show that these operators are finitary.

\subsubsection*{Related and future work.}
Our first interest in the semantics of datatypes in finiteness spaces was the
possibility of extending the quantitative semantics of the simply typed
$\lambda$-calculus in vectorial finiteness spaces to functional programming
with base datatypes. This would broaden the scope of the already well developped
proof theory of differential linear logic: the quantitative semantics provides
more precise information on cut elimination, and is thus a better guide in the
design of syntax than the plain relational interpretation. Earlier achievements
in this direction include the extension of the algebraic $\lambda$-calculus
\citep{vaux:alglam} with a type of booleans, for which the first author
established a semantic characterization of total terms: this is moreover proved to be
complete on boolean functions \citep{tasson:totality}. In previous unpublished work, we also
proposed a quantitative semantics of tail recursive iteration. As we mentioned
before, this did not provide a semantics of system $T$, which prompted us to
investigate the general structure of standard datatypes in finiteness spaces.
In this regard, our present contribution is an important step.

Notice that another standard approach to recursive datatypes is to consider the
impredicative encoding of inductive datatypes in system $F$. In an unpublished
preliminary version of his paper on finiteness spaces, Ehrhard proposed
an interpretation of second order linear logic. This is based on a class of functors
which, in particular, are monotonic for the finiteness extension order. Notice
however that this does not provide a denotational semantics \emph{stricto sensu}: 
in general, the interpretation decreases under cut elimination. Moreover, 
the possibility of a quantitative semantics in this setting is not clear.

Other accounts of type fixpoints in linear logic include the system of linear logic
proof nets with recursion boxes of \citet{gimenez:these}, which allows to
interpret, e.g., PCF. As such this system can be seen as a graphical
syntax for general recursion. Along similar lines, \citet{fmsw:rec} have
proposed a system of interaction nets which models iteration on
recursive datatypes. In both cases, no particular denotational
semantics is considered. Let us also mention \citeauthor{bm:mumall}'s $\mu$MALL
\citeyearpar{bm:mumall} which replaces the exponential modalities of linear logic
with least and greatest fixpoints: less close to our contribution, this work is
mainly oriented towards proof search. It however introduces the system $\mu$LJ
of intuitionistic logic with fixpoints, for which \citet{clairambault:these} later proposed a
cut elimination procedure allowing to encode system $T$, together with a game
semantics accounting for typed recursion.

The notion of transport functor we use to describe how functors in $\Fin$ can
be derived from functors in $\Rel$ is similar to the categorical
characterization of container types as relators with membership, by
\citet{hdm:contcat}: relators are functors in $\Rel$ which are monotonic for
inclusion of relations; membership relations are particular lax natural
transformations associated with these functors. The hypotheses we consider on
the functors in $\Rel$ underlying transport functors in $\Fin$
are weaker than those on relators with membership. On the other hand, in order
to ensure the functoriality of transport in $\Fin$, we are led to refer to a
\definitive{shape} relation: this side condition is essential for some
important instances, such as linear logic exponentials.

The relationship we establish between $\Rel$ and $\Fin$ might profitably be
recast in a more general setting. At least the transport lemma can be
reformulated for coherence spaces rather than finiteness spaces. Further
results of the paper might follow as well, up to some local tweaking of the
definitions (\eg, that of shape relations). It is still unclear to us whether the
approach we developped is limited to the restricted setting of finiteness
spaces, coherence spaces and maybe other web-based models
(\citeauthor{ehrhard:hypercoherences}'s hypercoherences
\citeyearpar{ehrhard:hypercoherences}, \citeauthor{loader:totality}'s totality
spaces \citeyearpar{loader:totality}), or if it can be generalized
in the spirit of the glueing and orthogonality techniques
studied by \citet{hs:glueing}.

\subsubsection*{Outline of the paper and main results.}
In section \ref{section:rel}, we review the structure and properties of $\Rel$.
We establish the transport lemma in section \ref{section:transport}, and derive
the interpretations of the positive connectives of linear logic in $\Fin$ from
those in $\Rel$. Section \ref{section:continuity} introduces two orders on
finiteness spaces and associated properties. In particular we provide
sufficient conditions for the existence of fixpoints of functors. We moreover prove
these conditions are automatically satisfied by transport functors. The last
two sections are dedicated to the finitary relational semantics of
$\lambda$-calculi: we first recall the semantics of the simply typed
$\lambda$-calculus in section \ref{section:lambda}, and then detail the semantics of
recursive algebraic datatypes in section \ref{section:rec}.

\section{Sets and relations}

\label{section:rel}

\subsection{Notations}

We write $\nat$ for the set of all natural numbers.  Let $A$ and $B$ be sets.
We write $A\subseteq B$ if $A$ is a subset of $B$ (not necessarily a strict one), 
and $A\finSubseteq B$ if moreover $A$ is finite.
We write $\card A$ for the cardinality of $A$, $\powerset A$ for the powerset
of $A$ and $\PFin A$ for the set of all finite subsets of $A$. We identify
multisets of elements of $A$ with functions $A\funArrow \nat$. If $\mu$ is such a
multiset, we write $\support\mu$ for its support set
$\set{\alpha\in A\st \mu(\alpha)\not=0}$. A finite multiset is a
multiset with a finite support. We write $\MFin A$ for the set of all finite
multisets of elements of $A$. Whenever $(\fullEl \alpha n)\in A^n$, we write
$\fullMulSet\alpha n$ for the corresponding finite multiset: $\alpha\in
A\mapsto\card{\set{i\st \alpha_i=\alpha}}$. We also write 
$\card\fullMulSet\alpha n=n$ for the cardinality of multisets.
The empty multiset is $\emptyMulSet$ and we use the additive notation for multiset union, \ie{} $\mu+\mu':\alpha\in A\mapsto \mu(\alpha)+\mu'(\alpha)$.

Since we will often consider numerous notions associated with a fixed set, we
introduce the following typographic conventions: we will in general use latin
majuscules for reference sets (e.g. $A$), greek minuscules for their elements
(e.g. $\alpha,\alpha'\in A$), latin minuscules for subsets (e.g. $a\subseteq
A$), gothic majuscules for sets of subsets (e.g. $\fA\subseteq\powerset A$),
and script majuscules for finiteness spaces (e.g. $\cA=(A,\fA)$).  In general,
if $T$ is an operation on sets we derive the notations for elements, subsets,
\emph{etc.} of $TA$ from those for elements, subsets, \emph{etc.} of $A$ by
the use of various overscripts (e.g. $\ag\alpha\in\ag a\subseteq TA$). We reserve
overlining for multisets (e.g. $\ms \alpha=\fullMulSet\alpha
n\in\MFin A$).

We will also consider families of objects (sets, elements, finiteness spaces,
\emph{etc.}) and thus introduce the following conventions. Unless stated
otherwise, all families considered in the same context are based on a common
set of indices, say $I$. We then write e.g. $\fam A$ for the family $\genFam AiI$.
We moreover use generic notations for componentwise operations on families: for
instance if $\fam A$ and $\fam B$ are two families of sets, we may write
$\fam{A\union B}$ for $\family{A_i\union B_i}{i\in I}$, and $\fam{\powerset A}$
for $\family{\powerset{A_i}}{i\in I}$. We may also write, e.g., $\fam{A\subseteq
B}$ for $A_i\subseteq B_i$ for all $i\in I$.

Assume $\fam A$ is a family of sets.  We write $\prod \fam A$ for the cartesian
product of the $A_i$'s and $\sum \fam A$ for their coproduct ($I$-indexed
disjoint union): $\prod\fam A=\set{\fam \alpha \st \forall i\in I,\
\alpha_i\in A_i}$ and $\sum\fam A=\set{(i,\alpha)\st i\in I \land \alpha\in
A_i}$.  We may of course denote finite products and coproducts of sets as
usual, e.g. $A\times B$ and $A+B$: in that case we assume indices are natural
numbers starting from $1$, e.g.  $A+B=\set{(1,\alpha)\st\alpha\in
A}\union\set{(2,\beta)\st\beta\in B}$.

\subsection{The category of sets and relations}\label{subsec:rel}

Let $A$ and $B$ be sets and $f$ be a relation from $A$ to $B$: $f\subseteq A\times
B$. We then write $\relRev f$ for the transpose relation $\set{(\beta,\alpha)\in
B\times A\st (\alpha,\beta)\in f}$. For all
subset $a\subseteq A$, we write $\relAppl fa$ for the \definitive{direct image}
of $a$ by $f$: $\relAppl fa \eqdef \set{\beta\in B \st\exists\alpha\in a,\
(\alpha,\beta)\in f}$.  If $\alpha\in A$, we will also
write $\relAppl f\alpha$ for $\relAppl f{\set \alpha}$.  
We say that a relation $f$ is \definitive{quasi-functional} if $\relAppl f\alpha$
is finite for all $\alpha$. 
If $b\subseteq B$, we define the \definitive{division} of $b$ by $f$ as
$\relDiv fb\eqdef\set{\alpha\in A\st \relAppl f\alpha\subseteq b}$. This is the
greatest subset of $A$ that $f$ maps to a subset of $b$: $\relDiv
fb=\Union\set{a\subseteq A\st\relAppl fa\subseteq b}$. Notice that in general
$\relAppl{f}{\pars{\relDiv fb}}$ may be a strict subset of $b$, and 
$\relDiv{f}{\pars{\relAppl fa}}$ may be a strict superset of $a$.

When $f\subseteq A\times B$ and $g\subseteq B\times C$, we denote by
$\relComp{g}{f}$ their composite: $(\alpha,\gamma)\in\relComp{g}{f}$ iff there
exists $\beta\in B$ such that $(\alpha,\beta)\in f$ and $(\beta,\gamma)\in g$.
Notice that this definition does not actually depend on the \definitive{types} of
$f$ and $g$ (namely the pairs of sets $(A,B)$ and $(B,C)$, respectively).
The identity relation on $A$ is the diagonal:
$\natInst{\id}A=\set{(\alpha,\alpha)\st \alpha\in A}\subseteq A\times A$. 
\begin{proposition}
  \label{prop:adjoint_relations}
  Let $f\subseteq A\times B$ be a relation. Then:
  \begin{itemize}
    \item $f=\relComp v{\relRev u}$ where $u$ and $v$ are (graphs of) functions,
      namely $u=\set{((\alpha,\beta),\alpha)\st (\alpha,\beta)\in f}$ and
      $v=\set{((\alpha,\beta),\beta)\st (\alpha,\beta)\in f}$, \ie{} the
      projections from $f$ onto its domain and image respectively;
    \item $f$ is a function 
      from $A$ to $B$ iff there exists $g\subseteq B\times A$ such that
      $\relComp{f}{g}\subseteq \natInst{\id}B$ and
      $\natInst{\id}A\subseteq\relComp gf$, and we then have $g=\relRev f$.
  \end{itemize}
\end{proposition}

Equipped with the above relational composition, relations form a category $\Rel$
whose objects are sets. 
More precisely, morphisms in $\in\Rel(A,B)$ are triples $(A,B,f)$ such
that $f\subseteq A\times B$: we use this trick only to ensure that homsets in
$\Rel$ are pairwise disjoint, which is part of the definition of a category
(see, \eg, \citet{maclane:cwm}). Then the identity morphism on set $A$ is 
$(A,A,\natInst{\id}A)$ and, if $(A,B,f)\in\Rel(A,B)$ and $(B,C,g)\in\Rel(B,C)$
then their composite in $\Rel$ is $(A,C,\relComp gf)\in\Rel(A,C)$. 
Most of the time, we will abuse notations and identify morphisms in $\Rel$ with
the underlying relations, especially when types are irrelevant for the
discussion or clear from the context: we may then simply write, \eg, $f$ for
$(A,B,f)$.
It is however important to notice that the action of functors on relations may
in general depend on their types:
\begin{definition} 
  A \definitive{functor} $T$ in $\Rel$ is the data of a set $TA$ for all set
  $A$ and a relation $\natInst T{A,B}f\subseteq TA\times TB$ for all $f\subseteq
  A\times B$, so that $T(A,B,f)\eqdef(TA,TB,\natInst T{A,B}f)\in\Rel(TA,TB)$,
  preserving identities and composition:
  $\natInst T{A,A}\natInst{\id}A=\natInst{\id}{TA}$ for all set $A$, and
  $\natInst T{A,C}(\relComp gf) =\relComp{\natInst T{B,C}g}{\natInst T{A,B}f}$ for all relations
  $f\subseteq A\times B$ and $g\subseteq B\times C$.
  \definitive{Cofunctors} are defined similarly, except for being contravariant, 
  \ie{} $\natInst T{A,B}f\subseteq TB\times TA$ 
  and $\natInst T{A,C}(\relComp gf) =\relComp{\natInst T{A,B}f}{\natInst T{B,C}g}$ for all relations
  $f\subseteq A\times B$ and $g\subseteq B\times C$.
\end{definition}
The simplest example of a functor (resp. cofunctor) is the identity functor
(resp. the transpose functor $(\relRev\cdot)$), which is the identity on sets and maps every
morphism $(A,B,f)$ to itself (resp. to its transpose $(B,A,\relRev f)$).
For all the care we took in making this definition precise, most of the time we
will leave out the superscripts and simply write $Tf$ both for $\natInst
T{A,B}f$ and the associated morphism. Moreover, we will sometimes restrict our
study to classes of functors for which such scripts are not actually relevant:
we say a functor $T$ is \definitive{type blind} if $\natInst T{A,B}f=\natInst
T{A',B'}f$ for all $f\subseteq (A\times B)\inter(A'\times B')$.

A functor $T$ is \definitive{monotonic on sets} if $TA\subseteq TB$ for
all sets $A\subseteq B$. If moreover $\natInst T{A,B}\natInst\id
A=\natInst\id{TA}$ (resp. $\natInst T{B,A}\natInst\id A=\natInst\id{TA}$) for
all sets $A\subseteq B$, we say $T$ \definitive{preserves inclusions} (resp.
\definitive{reverse inclusions}).
\begin{lemma}
  \label{lemm:inclusion-preserving-functors}
  Let $T$ be a functor. Then $T$ preserves inclusions (resp. reverse
  inclusions) iff $T$ is monotonic on sets and, for all $A\subseteq A'$,
  $B\subseteq B'$ and $f\subseteq A\times B$, we have 
  $\natInst T{A,B}f=\relComp{\natInst T{A',B'}f}{\natInst\id{TA}}$ (resp.
  $\natInst T{A,B}f=\relComp{\natInst\id{TB}}{\natInst T{A',B'}f}$).
  Moreover $T$ preserves both inclusions and reverse inclusions iff 
  $T$ is type blind.
\end{lemma}
\begin{proof}
  Assume $T$ preserves inclusions, $A\subseteq A'$, $B\subseteq B'$ and
  $f\subseteq A\times B$. Then $\relComp{\natInst T{A',B'}f}{\natInst\id{TA}}
  =\relComp{\natInst T{A',B'}(\relComp{\natInst \id B}f)}{\natInst T{A,A'}\natInst \id A}
  =\relComp{\natInst T{B,B'}\natInst\id B}{\natInst T{A,B}(\relComp f\natInst\id A)}
  =\relComp{\natInst \id{TB}}{\natInst T{A,B}f}=\natInst T{A,B}f$.
  For the converse, take $A=A'=B\subseteq B'$ and $f=\natInst \id A$.
  The case of reverse inclusion preserving functors is similar. 
  We conclude since type blindness is just the conjunction of both characterizations. 
\end{proof}

A functor $T$ is called a \definitive{relator} \citep{backhouse-et-al:polrel}
if it is monotonic for
relation inclusion: $\natInst T{A,B} f\subseteq \natInst T{A,B}f'$ as soon as
$f\subseteq f'\subseteq A\times B$.
We say a functor $T$ is \definitive{symmetric} if $T\left(\relRev
f\right)=\relRev{\left(Tf\right)}$ for all $f$.
\begin{lemma}
  \label{lemm:implications-of-functor-properties}
  A type blind functor is always a relator.
  Moreover, every relator is a symmetric functor.
\end{lemma}
\begin{proof}
  Assume $T$ is type blind and $f\subseteq f'\subseteq A\times B$. By the first item of
  Proposition~\ref{prop:adjoint_relations}, we can write 
  $f=\relComp v{\relRev u}$ and $f'=\relComp {v'}{\relRev {u'}}$
  where $u$, $v$, $u'$, $v'$ are the graphs of functions 
  $\funDec uCA$, $\funDec vCB$, $\funDec {u'}{C'}A$ and 
  $\funDec {v'}{C'}B$ with $C\subseteq C'$ (in fact
  $C=f$ and $C'=f'$) and such that $u=\funRest {u'}C$ and $v=\funRest {v'}C$, 
  or equivalently $u=\relComp{u'}{\natInst\id C}$ and $v=\relComp{v'}{\natInst\id C}$.
  Since $T$ is type blind, $TC\subseteq TC'$ and we can write: $Tf
  =T(\relComp v{\relRev u})
  =\relComp{Tv}{T\relRev u}
  =\relComp{T(\relComp{v'}{\natInst\id C})}{T(\relComp{\natInst\id C}{\relRev{u'}})}
  =\relComp{T{v'}}{\relComp{T\natInst\id C}{\relComp{T\natInst\id C}{T\relRev{u'}}}}
  \stackrel{(*)}{=}\relComp{T{v'}}{\relComp{\natInst\id{TC}}{T\relRev{u'}}}
  \subseteq\relComp{T{v'}}{\relComp{\natInst\id{TC'}}{T\relRev{u'}}}
  =\relComp{T{v'}}{T\relRev{u'}}
  =T f'$.
  The crucial step is $(*)$, which refers to
  Lemma~\ref{lemm:inclusion-preserving-functors}:
  $\natInst T{C,C'}\natInst\id C=\natInst T{C',C}\natInst\id C=\natInst\id{TC}$.

  Now assume $T$ is any relator. Then by the second item of 
  Proposition~\ref{prop:adjoint_relations}, $T\left(\relRev
  f\right)=\relRev{\left(Tf\right)}$ as soon as $f$ is the graph of a function.
  This extends to all relations by the first item of 
  Proposition~\ref{prop:adjoint_relations}.
\end{proof}
In particular, every inclusion preserving symmetric functor is a
relator since it is type blind.\footnote{This fixes a flawed result by
\citet[Lemma 5.1]{bdm:aop}, which implicitly relies on every functor preserving
inclusions.}
Of course, not all functors are type blind (resp. relators, symmetric):
\begin{counter}
  Let $P$ denote the functor of powersets and direct images: 
  $PA=\powerset A$ and $\natInst P{A,B}f=\set{(a,\relAppl fa)\st a\subseteq A}$.
  Notice that $\natInst P{A,B}f$ is the graph of a function, which 
  is not necessarily injective, hence $P$ is not symmetric. 
  By the previous lemma, $P$ is neither type blind nor a relator.
\end{counter}

A functor is said to be \definitive{continuous on sets} if it
preserves directed unions of sets: $T\Union\fam A=\Union
\fam{TA}$ for all family of sets $\fam A$ which is directed
for inclusion. Notice the use of our convention for denoting families:
$\fam{TA}=\family{TA_i}{i\in I}$.
Similarly, we say $T$ is \definitive{continuous on relations} if it preserves
directed unions of parallel relations: $T\Union\fam f=\Union
\fam{Tf}\in\Rel(TA,TB)$ for all family of relations $\fam f\in\Rel(A,B)^I$
which is directed for inclusion.
\begin{lemma}
  \label{lemm:blind-continuous}
  If $T$ is continuous on sets (resp. on relations) then it is monotonic on
  sets (resp. it is a relator). Moreover, if $T$ is type blind and continuous
  on relations then it is also continuous on sets.
\end{lemma}
\begin{proof}
  That continuity implies monotonicity is standard.
  Assume $T$ is type blind and continuous on relations, 
  and let $\fam A$ be a directed family of sets. 
  Then $\natInst\id{T\Union\fam A}
  =T\natInst\id{\Union\fam A}
  =T\Union\fam{\natInst\id A}
  =\Union\fam{T\natInst\id A}
  =\Union\fam{\natInst\id{TA}}
  =\natInst\id{\Union\fam{TA}}$
  hence $T\Union\fam A=\Union\fam{TA}$.
\end{proof}
Notice again our use of the arrow notation for families. This allows to keep
our developments concise while remaining self-explanatory and
unambiguous: here $\fam A$ is a family of sets, while $T$ is just one functor,
hence, \eg, $\fam{\natInst\id{TA}}$ can only mean
$\family{\natInst\id{TA_i}}{i\in I}$. Although this needs some overhead effort
to parse the first times, we are confident the reader will quickly become
familiar with this convention: we will rely on its conciseness extensively in
the remaining of the paper, always taking care that it does not introduce any
ambiguity.

If $T$ is continuous on both sets and relations, then we simply say it is
\definitive{continuous}. Of course, the identity functor is a type blind
continuous functor.
Another standard example is the multiset functor given by $\oc A=\MFin
A$ and, for all relation $f$, \[\oc f=\set{(\fullMulSet\alpha
n,\fullMulSet\beta n)\st n\in\nat\land \forall k,\ (\alpha_k,\beta_k)\in f}.\]
When $a\subseteq A$, we write $\prom a=\MFin a\subseteq \oc A$ (rather that
$\oc a\subseteq\oc A$) in order to avoid confusion with the corresponding operation
on relations.

Let $T$ and $U$ be two functors from $\Rel$ to $\Rel$, and let $f$ be the data
of a relation $\natInst f{A}$ from $TA$ to $UA$ for all set $A$: we say $f$ is
a \definitive{lax natural transformation} from $T$ to $U$ if, for all relation
$g$ from $A$ to $B$, $\relComp{\natInst
fB}{\pars{Tg}}\subseteq\relComp{\pars{Ug}}{\natInst fA}$. 
We say $f$ is a \definitive{natural transformation} if this inclusion is
always an equality.  In general we omit
the annotation and simply write $f$ for $\natInst fA$ when $A$ is clear from
the context.  Of course, the identities $\natInst{\id}{TA}$ constitute a
natural transformation from each $T$ to itself.
For all set $A$, consider the only relation $\supp$ from
$\oc{A}$ to $A$ such that $\relAppl{\supp}{\ms\alpha}=\support{\ms\alpha}$ for
all $\ms\alpha\in\oc{A}$. This defines a lax natural transformation from $\oc$ to
the identity functor: notice that in that case, the inclusion 
$\relComp{\supp}{\oc g}\subseteq \relComp g{\supp}$ may be strict.
Lax natural transformations between type blind functors enjoy some kind of
stability property:
\begin{lemma}
  \label{lemma:natural:restriction}
  Let $T$ and $U$ be type blind functors and 
  let $f$ be a lax natural transformation from $T$ to $U$
  Then, if $A\subseteq B$:
  \begin{itemize}
    \item $\natInst fA=\relComp{\natInst fB}{\natInst{\id}{TA}}$;
    \item for all $\ag a\subseteq TA$, 
      $\relAppl{\natInst fA}{\ag a}=\relAppl{\natInst fB}{\ag a}$;
    \item for all $\widehat b\subseteq UB$, 
      $\relDiv{\natInst fA}{\widehat b}=\pars{\relDiv{\natInst
fB}{\widehat b}}\, \inter TA$.
  \end{itemize}
\end{lemma}
\begin{proof}
  By applying the naturality condition to the identity $\natInst{\id}A$
  both as a relation from $A$ to $B$ and as a relation from $B$ to $A$, 
  we obtain
  $\relComp{\natInst fB}{\natInst{\id}{TA}}\subseteq\relComp{\natInst{\id}{UA}}{\natInst fA}$
  and 
  $\relComp{\natInst fA}{\natInst{\id}{TA}}\subseteq\relComp{\natInst{\id}{UA}}{\natInst fB}$, 
  hence $\natInst fA=\relComp{\natInst fB}{\natInst{\id}{TA}}=\natInst fB\inter(TA\times UB)$, 
  from which both other properties follow.
\end{proof}

We shall not restrict our study to unary functors, hence we need to generalize
the above notions to families of relations indexed by a fixed set $I$. 
If $\fam A$ and $\fam B$ are families of sets, we call \definitive{relation
from $\fam A$ to $\fam B$} any family $\fam f$ of componentwise relations:
$\fam{f\subseteq A\times B}$, \ie{} for all $i\in I$, $f_i\subseteq A_i\times
B_i$.
We denote by $\Rel^I$ the category of $I$-indexed families of sets and
relations, with componentwise identities and composition: $\natInst\id{\fam
A}=\fam{\natInst\id A}$ and $\relComp{\fam g}{\fam f}=\fam{\relComp gf}$.
Again, we must precise that morphisms in $\Rel^I(\fam A,\fam B)$ are triples
$(\fam A,\fam B,\fam f)$ such that $\fam f$ is a relation from $\fam A$ to
$\fam B$, although we simply write $\fam f$ for $(\fam A,\fam B,\fam f)$
whenever $\fam A$ and $\fam B$ are clear from the context.

An $I$-ary functor in $\Rel$ is a functor from $\Rel^I$ to $\Rel$, \ie{} the
data of a set $T\fam A$ for all $I$-indexed family $\fam A$ of sets, and of a
relation $\natInst T{\fam A,\fam B}\fam f$ from $T\fam A$ to $T\fam B$ for all
relation $\fam f$ from $\fam A$ to $\fam B$, preserving identities and
composition: 
$\natInst T{\fam A,\fam A}\natInst\id{\fam
A}=\natInst\id{T\fam A}$ and $\natInst T{\fam A,\fam C}\fam{\relComp gf}
=\relComp{\natInst T{\fam B,\fam C}\fam g}{\natInst T{\fam A,\fam B}\fam f}$.
Whenever $\fam A$ and $\fam B$ are clear from the context, we just write
$T\fam f$ for both $\natInst T{\fam A,\fam B}\fam f$ and $T(\fam A,\fam B,\fam
f)=(T\fam A,T\fam B,\natInst T{\fam A,\fam B}\fam f)$.

Let $T$ be an $I$-ary functor. We say:
\begin{itemize}
  \item $T$ is \definitive{type blind} if $\natInst T{\fam A,\fam B}\fam f=\natInst
    T{\fam{A'},\fam{B'}}\fam f$ whenever both sides of the equation are defined;
  \item $T$ is \definitive{monotonic on sets} if $T\fam A\subseteq T\fam B$ for all families
    $\fam A$ and $\fam B$ of sets such that $\fam{A\subseteq B}$;
  \item $T$ is an \definitive{$I$-ary relator} if $T\fam f\subseteq T\fam g$ for all
    relations $\fam f$ and $\fam g$ from $\fam A$ to $\fam B$ such that
    $\fam{f\subseteq g}$;
  \item $T$ is \definitive{symmetric} if $T\fam{\relRev f}=\relRev{\pars{T\fam f}}$ for all
    family of relations $\fam f$.
\end{itemize}
Lemmas \ref{lemm:inclusion-preserving-functors} and
\ref{lemm:implications-of-functor-properties} extend to $I$-ary functors:
$T$ is type blind iff $T$ preserves $I$-ary inclusions $\natInst\id{\fam A}$,
both from $\fam A$ to $\fam B$ and from $\fam B$ to $\fam A$, for all
$\fam{A\subseteq B}$; every type blind functor is monotonic on sets and is a
relator; every relator is symmetric.

In order to define the continuity of $I$-ary functors, we have to consider
families of families. We thus introduce the following conventions: 
by $\fam{\maf{A}}$, we denote an
$I$-indexed family $\genFam{\maf{A}}iI$ of families of sets,
where each $\maf{A}_i=\family{A_{i,j}}{j\in J_i}$ takes indices in some
variable set $J_i$. If $\fam {j\in J}$ (i.e. $j_i\in J_i$ for all $i\in I$), 
we also write $\fam A_{\fam j}$ for the $I$-indexed family $\family{A_{i,j_i}}{i\in I}$.
We use leftwards arrows to distinguish families indexed by some variable set
from $I$-indexed families.
When the order of application of arrows is reversed, as in $\maf{\fam A}$, the
leftwards arrow stands for quantifying over all families $\fam j\in\prod{\fam
J}$ of indices, \ie{} $\maf{\fam A}=\family{\fam A_{\fam j}}{\fam{j\in
J}}$.

We say $\fam{\maf{A}}$ is \definitive{directed} if each
$\maf{A}_i$ is directed for inclusion. The family $\fam{\Union\maf{A}}=
\family{\Union\maf{A}_i}{i\in I}$ is 
the componentwise union of $\fam{\maf{A}}$.
Then we say $T$ is \definitive{continuous on sets} if it commutes to directed unions: 
$T\fam{\Union\maf A}=\Union\maf{T\fam A}$ as soon as $\fam{\maf A}$ is directed.
Similarly, we say $T$ is \definitive{continuous on relations} if, for all
directed family $\fam{\maf f}$, with $f_{i,j}\subseteq A_i\times B_i$ for all
$i\in I$ and $j\in J_i$, we have $T\fam{\Union\maf f}=\Union\maf{T\fam f}$.
If both properties hold, we simply say $T$ is continuous.
Again, Lemma \ref{lemm:blind-continuous} extends to $I$-ary functors: every
type blind functor which is continuous on relations is continuous.

We denote by $\Pi_i$ the $i$-th \definitive{projection functor} from $\Rel^I$
to $\Rel$: for all family of sets $\fam A$, $\Pi_i\fam A=A_i$ and, for all
relation $\fam f$ from $\fam A$ to $\fam B$, $\Pi_i\fam f=f_i$.
Projection functors are continuous type blind relators. Other standard examples include:
the \definitive{cartesian product} functor, given by $\Tensor\fam A=\prod \fam
A$ and $\Tensor\fam f=\set{\pars{\fam\alpha,\fam\beta}\st\fam{(\alpha,\beta)\in f}}$; 
and the \definitive{disjoint union} functor,
given by $\Oplus\fam A=\sum\fam A$ and $\Oplus\fam
f=\set{\pars{(i,\alpha),(i,\beta)}\st i\in I\land (\alpha,\beta)\in f_i}$.
Notice that $\Oplus$ defines both products and coproducts in $\Rel$:
we may also write it $\With$ when we refer to it as the functor of 
products.

Let $T$ and $U$ be two functors from $\Rel^I$ to $\Rel$, and let $f$ be the data of a
relation $\natInst f{\fam A}$ from $T\fam A$ to $U\fam A$ for all $\fam A$:
we say $f$ is a \definitive{lax natural transformation} from $T$ to $U$, 
if, for all relation $\fam g$ from $\fam A$ to $\fam B$, 
$\relComp{\natInst f{\fam B}}{\pars{T\fam g}}\subseteq\relComp{\pars{U\fam
g}}{\natInst f{\fam A}}$. We say $f$ is a \definitive{natural transformation}
if moreover this inclusion is always an equality. Again, the identities $\natInst\id{T\fam A}$
define a natural transformation from $T$ to itself. Other basic
examples of natural transformations are the following
\definitive{projection}, \definitive{restriction} and 
\definitive{index} relations:
\begin{itemize}
  \item for all $i\in I$, the projection from $\Tensor\fam A$ to $A_i$ is
    $\proj_i=\set{(\fam\alpha,\alpha_i)\st\fam\alpha\in\Tensor\fam A}$;
  \item for all $i\in I$, the restriction from $\Oplus\fam A$ to $A_i$ is 
    $\rest_i=\set{((i,\alpha),\alpha)\st\alpha\in A_i}$;
  \item the index relation $\indx$ from $\Oplus\fam A$ to $I$ is given by
    $\indx=\set{((i,\alpha),i)\st i\in I\land \alpha\in A_i}$. 
\end{itemize}
Then: each $\proj_i$ is a natural transformation from $\Tensor$ to $\Pi_i$;
each $\rest_i$ is a natural transformation from $\Oplus$ to $\Pi_i$;
and $\indx$ is a natural transformation from $\Oplus$ to $E_I$,
which is the constant functor $E_I\fam A=I$ and $E_I\fam f=\natInst{\id}I$.
Again, Lemma~\ref{lemma:natural:restriction} extends to $I$-ary type blind
functors and lax natural transformations between them.

\section{On the transport of finiteness structures}

\label{section:transport}

\subsection{Finiteness spaces}

Let $A$ and $B$ be sets, we write $A \finPolar B$ if $A\inter B$ is finite.  If
$\fA\subseteq\powerset A$, we define the \definitive{predual} of $\fA$ on $A$ as
$\dual[A]{\fA}=\set{a'\subseteq A\st \forall a\in \fA,\ a\finPolar a'}$.  By 
standard arguments on closure operators and orthogonality constructions, we have the
following properties:
\begin{itemize}
\item $\PFin A\subseteq\dual[A]\fA$; 
\item $\fA\subseteq\bidual[A]{\fA}$; 
\item if $\fA\subseteq\fA'$, then $\dual[A]{\fA'}\subseteq\dual[A]{\fA}$ and
  $\bidual[A]{\fA}\subseteq\bidual[A]{\fA'}$;
\item by the previous two items, $\dual[A]{\fA}= \tridual[A]{\fA}$;
\item $\dual[A]\fA$ is downwards closed for inclusion, \ie{} 
  $a\subseteq a'\in\dual[A]\fA$ implies $a\in\dual[A]\fA$;
\item $\dual[A]\fA$ is closed under finite unions, \ie{} 
  $a,a'\in\dual[A]\fA$ implies $a\union a'\in\dual[A]\fA$;
\item if $\bidual[A]\fA_i=\fA_i$ for all $i\in I$, then
  $\bidual[A]{\left(\Inter\fam\fA\right)}=\Inter\fam\fA$.
 \end{itemize}
A \definitive{finiteness structure} on $A$
is a set $\fA$ of subsets of $A$ such that $\bidual[A]{\fA} = \fA$.  Then a
\definitive{finiteness space} is a pair $\cA=\pars{\web\cA,\fin\cA}$
where $\web\cA$ is the underlying set, called the \definitive{web} of $\cA$,
and $\fin \cA$ is a finiteness structure on $\web{\cA}$.  We write $\dual{\cA}$
for the \definitive{dual} finiteness space: $\web{\dual{\cA}}=\web{\cA}$ and
$\fin{\dual{\cA}}=\dual[\web{\cA}]{\fin{\cA}}$.  The elements of $\fin{\cA}$
are called the \definitive{finitary subsets} of $\cA$.

For every set $A$, $(A,\PFin A)$ is a finiteness space
and $\dual{(A,\PFin A)}=(A,\powerset A)$.
In particular, each finite set $A$ is the web of exactly one
finiteness space: $(A,\PFin A)=(A,\powerset A)$. We introduce
the empty finiteness space $\fsEmpty$ with web $\emptyset$ and
the singleton finiteness space $\fsSgl$ with web $\set\emptyset$.
Having finite webs, $\fsEmpty$ and $\fsSgl$ are
identified with their respective duals:
$\fsZero=\dual\fsEmpty=\fsEmpty$ and $\fsBot=\dual\fsSgl=\fsSgl$.
We moreover introduce the space of \emph{flat natural numbers} $\fsFlatNat
=\pars{\nat,\PFin{\nat}}$.

The following reformulation of bidual closure is given by \citet{ehrhard:fs}:
\begin{lemma}
If $\fA\subseteq\powerset A$ is downwards closed for inclusion, 
then $a\in\bidual[A]\fA$ iff, for all infinite subset $a'\subseteq a$,
there is an infinite subset $a''\subseteq a'$  such that $a''\in\fA$.
\end{lemma}
In particular, the following \emph{does not} define a finiteness structure:
\begin{counter}[Communicated to us by Laurent Regnier]
  We say $t\subseteq\nat$ is \definitive{thin} if the sequence
  $\family{\frac{\card{t\inter\zeroTo {n-1}}}n}{n\in\nat}$ converges to $0$.
  Let $\thinSubsets$ be the set of all thin subsets of $\nat$. Examples of
  infinite thin subsets are $\set{n^2\st n\in\nat}$ and $\set{n^n\st n\in\nat}$.
  Of course, $\nat$ itself is not thin. Notice that every infinite subset
  $a\subseteq\nat$ contains an infinite thin subset: let $\genFam{\alpha}n\nat$ be the 
  ordered sequence of the elements of $a$; then, for instance,
  $\set{\alpha_{n^2}\st  n\in\nat}\in\thinSubsets$. Notice moreover that
  $\PFin{\nat}\subseteq\thinSubsets$, and that $\thinSubsets$ is downwards closed
  for inclusion and closed under finite unions. By the previous lemma, 
  $\nat\in\bidual[A]{\thinSubsets}$ and then
  $\bidual[\nat]{\thinSubsets}=\powerset{\nat}\not=\thinSubsets$. 
\end{counter}

All along the text, we provide relevant counter-examples in order to
motivate the various notions we introduce, and also to emphasize the complex
structure of finiteness spaces. These will often refer to a situation like the
above one: we say $\fA\subseteq\powerset A$ is a \definitive{fake
finiteness structure} on $A$ if $\fA$ is downwards closed for inclusion, 
closed under finite unions, and contains $\PFin{A}$, but
$\fA\not=\bidual[A]\fA$. Below we present another fake finiteness
structure, the properties of which will be useful in some of our arguments.

\begin{counter}
  \label{counter:codaggers}
  For all $n\in\nat$, write $\dagger_n=\set{(p,q)\st p=n\lor q=n}$.
  Then, for all $n\in\nat$, write $\fC_n=\dual[\nat\times\nat]{
  \set{\dagger_p\st p\ge n}}$. Being a dual set, each $\fC_n$ is a finiteness
  structure on $\nat\times\nat$. Moreover, $\fC_n\subseteq\fC_{n'}$ as soon as
  $n\le n'$. As a consequence, $\fC=\Union\fam\fC$ is downwards
  closed for inclusion, closed under finite unions and contains all finite
  subsets, but not every subset. However,
  $\dual[\nat\times\nat]{\fC}=\PFin{\nat\times\nat}$ whose dual is $\powerset{\nat\times\nat}$.
\end{counter}

\subsection{Transport of finiteness structures}
The following lemma will be used throughout the paper. It allows to transport a
finiteness structure on set $B$, along any relation $f$ from $A$ to $B$, provided
$f$ maps finite subsets of $A$ to finitary subsets of $B$.
\begin{lemma}[Transport]\label{lemma:transport}
  Let $A$ be a set, $\cB$ a finiteness space and 
  $f$ a relation from $A$ to $\web \cB$ such that
  $\relAppl f\alpha\in\fin\cB$ for all $\alpha\in A$.\footnote{
  Following the terminology of \citet{hs:glueing},
  this condition can be rephrased as $f$ being negative from 
  $(A,\powerset A,\PFin A)$ to $(\web\cB,\fin\cB,\fin{\dual\cB})$, 
  \ie{} for all $a\subseteq A$ and $b'\in\fin{\dual\cB}$, 
  $a\finPolar \relAppl{\relRev f} b'$ implies $\relAppl f a\polar b'$.
  It is however unclear, at the time of writing, under which hypotheses the transport lemma could be
  recast in this more general setting.
  }
  Then  $\transport{\cB,f}\eqdef \set{a \subseteq A
	\st\relAppl fa\in\fin\cB}$
  is a finiteness structure on $A$ and, more precisely, 
  $\transport{\cB,f}=\bidual[A]{\set{\relDiv fb\st b\in\fin{\cB}}}$.
\end{lemma}
\begin{proof}
  Write $\fA=\set{\relDiv fb\st b\in\fin\cB}$.
  The first inclusion is easy: $\transport{\cB,f}\subseteq\bidual[A]{\fA}$
  because, for all $a\in\transport{\cB,f}$ and $a'\in \dual[A]{\fA}$,
  $a\inter a'$ is finite. Indeed, $\relAppl fa\in\fin\cB$
  hence $a'\inter\pars{\relDiv f{\pars{\relAppl fa}}}$ is finite; 
  moreover $a\subseteq\relDiv f{\pars{\relAppl fa}}$.

  We now prove the reverse inclusion: let $a\in \bidual[A]{\fA}$, we establish
  that $a\in \transport{\cB,f}$, \ie{} $\relAppl fa\in\fin\cB$. It is sufficient to show
  that, for all $b'\in\fin{\dual{\cB}}$, $b''=\pars{\relAppl fa}\inter b'$ is
  finite. Since $b''\subseteq \relAppl fa$, for all $\beta\in b''$ there is
  $\alpha\in a$ such that $\beta\in\relAppl f\alpha$: by the axiom of choice,
  we obtain a function $\funDec{\phi}{b''}{a}$ such that $\beta\in\relAppl
  f{\phi(\beta)}$ for all $\beta\in b''$, which entails $b''\subseteq\relAppl
  f{\phi(b'')}$. Now it is sufficient to show that $\phi(b'')$ is finite.
  Indeed, in that case, $\relAppl
  f{\phi(b'')}=\Union_{\alpha\in\phi(b'')}\relAppl f\alpha$ is a finite union
  of finitary subsets of $\cB$: recall that by our hypothesis on $f$, $\relAppl
  f\alpha\in\fin{\cB}$ for all $\alpha\in A$. Hence $b''\in\fin{\cB}$ and,
  since we also have $b''\subseteq b'\in\fin{\dual{\cB}}$, $b''$ is finite.

  Since $\phi(b'')\subseteq a\in\bidual[A]{\fA}$, it will be sufficient to prove
  that $\phi(b'')\in \dual[A]{\fA}$ also. For that purpose, we consider
  $b\in\fin\cB$ and prove that $a''=\phi(b'')\inter\relDiv fb$ is finite.  If
  $\alpha\in a''$, there exists $\beta\in b''$ such that $\alpha=\phi(\beta)$
  and moreover $\relAppl f\alpha\subseteq b$; since $\beta\in\relAppl
  f{\phi(\beta)}=\relAppl f\alpha$, we obtain that $\beta\in b''\inter b$.
  Hence $a''\subseteq \phi(b''\inter b)$, which is finite because $\phi$ is a function and 
  $b''\inter b\subseteq b'\inter b$ is finite as $b'\in\fin{\dual{\cB}}$ and $b\in\fin{\cB}$.
\end{proof}
The reader should remark that the structure of this proof is very similar
to that of the characterization of the exponential modality, given by
\citet[Lemma 4]{ehrhard:fs}. Actually, we obtain this characterization as 
a straightforward application of transport:
\begin{example}\label{ex:oc}
  Let $\cA=(A,\fA)$ be a finiteness space, and recall that $\natInst\supp A$ is the only relation
  from $\oc{A}$ to $A$ such that $\relAppl{\natInst\supp A}{\ms
  \alpha}= \support{\ms\alpha}$ for all $\ms\alpha\in\oc{A}$. Notice in
  particular that $\support{\ms\alpha}\in\PFin{A}\subseteq \fA$.
  By the transport lemma, $(\oc A, \transport{\cA,\natInst\supp A})$ is a finiteness space that
  we denote by $\oc\cA$. We moreover have that $\relDiv{\natInst\supp A}a=\MFin{a}=\prom a$, and
  we obtain:
  \begin{equation*}
    \fin{\oc\cA}
    =\set{\ms a\subseteq\oc{\web\cA}\st\relAppl{\natInst{\supp}{\web{\cA}}}{\ms a}\in\fin\cA}
    =\bidual[\web{\cA}]{\set{\prom{a}\st{}a\in\fin{\cA}}}.
  \end{equation*}
\end{example}

The transport lemma is easily generalized to families of finiteness structures.
If we write $\relDiv{\fam f}{\fam b}$ for $\Inter\fam{\relDiv fb}=\Inter_{i\in
I}\pars{\relDiv{f_i}{b_i}}$, we obtain:
\begin{corollary}
  \label{coro:transport:family}
  Let $A$ be a set, $\fam{\cB}$ a family of finiteness spaces
  and $\fam f$ a family of relations such that, for all $\alpha\in
  A$ and all $i\in I$, $\relAppl {f_i}\alpha\in\fin{\cB_i}$.
  Then  $\transport{\fam{\cB},\fam f} \eqdef \set{a \subseteq A
  \st \forall i\in I,\ \relAppl{f_i}a\in\fin{\cB_i}}$
  is a finiteness structure on $A$ and, more precisely,
  $\transport{\fam{\cB},\fam f} = \bidual[A]{\set{\relDiv{\fam f}{\fam b}
  \st \fam{b\in\fin{\cB}}}}$.
\end{corollary}
\begin{proof}
  By Lemma~\ref{lemma:transport}, each $\transport{\cB_i,f_i}$ is a finiteness
  structure on $A$. 
  As bidual closure commutes to intersections of finiteness structures,
  $\transport{\fam{\cB},\fam f}=\Inter_{i\in I}\transport{\cB_i,f_i}$
  is a finiteness structure.
  Let us prove that $\transport{\fam{\cB},\fam f} = \bidual[A]{\set{\Inter_{i\in
  I}\pars{\relDiv{f_i}{b_i}}  \st \fam{b\in\fin{\cB}}}}$.
  Let $a\in\transport{\fam{\cB},\fam f}$: for all $i\in I$, 
  $a\in\transport{\cB_i,f_i}$, hence setting $b_i=\relAppl{f_i}{a}$ we obtain 
  $b_i\in\fin{\cA_i}$ and $a\subseteq\relDiv{f_i}{b_i}$. We have thus found 
  $\fam{b\in\fin{\cB}}$ such that $a\subseteq\Inter_{i\in I}\pars{\relDiv{f_i}{b_i}}$, 
  which proves one inclusion.
  For the reverse, let $\fam {b\in \fin{\cB}}$: for all $j\in I$, 
  $\Inter_{i\in I}(\relDiv{f_i}{b_i})\subseteq\relDiv{f_j}{b_j}$. Now, observe 
  that $\relDiv{f_j}{b_j}\in \transport{\cB_j,f_j}$ which is downwards closed for inclusion,
  hence $\Inter_{i\in I}(\relDiv{f_i}{b_i})\in\transport{\cB_j,f_j}$.
  We have just proved that $\set{\relDiv{\fam f}{\fam b}
  \st \fam{b\in\fin{\cB}}}\subseteq \transport{\fam{\cB},\fam f}$, 
  and we conclude since
  bidual closure is monotonic and idempotent.
\end{proof}

\begin{example}
  For all family $\fam{\cA}$ of finiteness spaces, 
  we denote by $\Tensor\fam{\cA}$ the finiteness space 
  $\pars{\prod\fam{\web{\cA}},\transport{\fam{\cA},\fam\proj}}$:
  for all $\ag a\subseteq\prod\fam{\web{\cA}}$, 
  $\ag a\in\fin{\Tensor\fam{\cA}}$ iff 
  $\relAppl{\proj_i}{\ag a}\in\fin{\cA_i}$ for all $i\in I$.
  We moreover obtain $\fin{\Tensor\fam{\cA}}=\bidual[\prod\fam{\web{\cA}}]{
  \set{\prod \fam a\st \fam{a\in\fin{\cA}}}}$.

  Similarly, let $\With\fam{\cA}$ be the finiteness space
  $\pars{\sum\fam{\web{\cA}},\transport{\fam{\cA},\fam\rest}}$: for all $\ag
  a\subseteq\sum\fam{\web{\cA}}$, $\ag a\in\fin{\With\fam{\cA}}$ iff
  $\relAppl{\rest_i}{\ag a}\in\fin{\cA_i}$ for all $i\in I$.  Notice that this implies
  $\fin{\With\fam{\cA}}={\set{\sum\fam a\st \fam{a\in\fin{\cA}}}}$, hence the
  bidual closure is optional in that case.

  Finally, we define the finiteness space $\Oplus\fam{\cA}=
  \pars{\sum\fam{\web{\cA}},\transport{(\fam{\cA},\cI),(\fam\rest,\indx)}}$
	where $\cI=(I,\PFin{I})$:
  $\ag a\in\fin{\Oplus\fam{\cA}}$ iff $\relAppl{\indx}{\ag a}$ is finite and 
  $\relAppl{\rest_i}{\ag a}\in\fin{\cA_i}$ for all $i\in I$.
  We obtain $\fin{\Oplus\fam{\cA}}=
  \set{\sum_{i\in J}a_i\st J\finSubseteq I\land \forall i\in J,\ a\in\fin{\cA_i}}$, 
  the bidual closure being optional.
  We have $\parDual{\Oplus\fam{\cA}}=\With{\fam{\dual{\cA}}}$,
  and moreover $\Oplus\fam{\cA}=\With\fam{\cA}$ when $I$ is finite.
\end{example}

Finally we introduce two other constructions on finiteness spaces which are
not directly obtained by transport. If $\fam\cA$ is a family of finiteness
spaces, we set $\Parr\fam\cA=\parDual{\Tensor\fam{\dual\cA}}$. From this, 
we derive $\cA\limpl\cB=\dual\cA\parr\cB=\parDual{\cA\tensor\dual\cB}$
for all finiteness spaces $\cA$ and $\cB$.

\subsection{Finitary relations}

Let $\cA$ and $\cB$ be two finiteness spaces: we say a relation $f$ from
$\web{\cA}$ to $\web{\cB}$ is \definitive{finitary} from $\cA$ to $\cB$ if: for all
$a\in\fin{\cA}$, $\relAppl fa\in\fin{\cB}$, and for all
$b'\in\fin{\dual{\cB}}$, $\relAppl{\relRev f}{b'}\in\fin{\dual{\cA}}$.
The following characterization of finitary relations is given by
\citet[Section 1.1]{ehrhard:fs}:

\begin{lemma}\label{lemma:finitary:relation}
  Let $f\subseteq\web{\cA}\times\web{\cB}$. The following propositions 
  are equivalent:
  \begin{enumerate}[(a)]
    \item \label{fin:def} $f$ is finitary from $\cA$ to $\cB$;
    \item \label{fin:rev} $\relRev f$ is finitary from $\dual{\cB}$ to $\dual{\cA}$;
    \item \label{fin:point} for all $a\in\fin{\cA}$, 
      $\relAppl fa\in\fin{\cB}$ and, for all $\beta\in\web{\cB}$, 
      $\relAppl{\relRev f}\beta\in\fin{\dual\cA}$;
    \item \label{fin:limpl} $f\in\fin{\cA\limpl\cB}$. 
  \end{enumerate}
\end{lemma}
Notice that the identity relation $\natInst\id{\web\cA}$ is finitary from
$\cA$ to itself, and that finitary relations compose: we thus
introduce the category $\Fin$ whose objects are finiteness spaces and morphisms
are finitary relations. Again, although we should precise that a morphism in
$\Fin(\cA,\cB)$ is a triple $(\cA,\cB,f)$ such that $f$ is a finitary relation
from $\cA$ to $\cB$, we will in general identify $\Fin(\cA,\cB)$ with
$\fin{\cA\limpl\cB}$.
Functors in $\Fin$ are defined similarly to those in $\Rel$: a functor $\cT$ is
the data of a finiteness space $\cT\cA$ for all finiteness space $\cA$ and of a
finitary relation $\natInst\cT{\cA,\cB}f$ from $\cT\cA$ to $\cT\cB$ for all
$f\in\Fin\pars{\cA,\cB}$, preserving identities and composition.

Some functors in $\Fin$ give rise to functors in $\Rel$: we say a
functor $\cT$ in $\Fin$ \definitive{has a web} if there exists a
functor $T$ in $\Rel$, such that $\web{\cT\cA}=T{\web\cA}$ for all
finiteness space $\cA$, and $\natInst\cT{\cA,\cB}
f=\natInst T{\web\cA,\web\cB}f$ for all $f\in\Fin\pars{\cA,\cB}$. 
We then say $T$ is the web of $\cT$ and
write $T=\web\cT$. Notice that in that case, $\natInst T{A,B} f$ must be
finitary from $\cT{(A,\fA)}$ to $\cT{(B,\fB)}$ whenever $f$ is
finitary from ${(A,\fA)}$ to ${(B,\fB)}$.
We say $\cT$ is type blind if $\natInst\cT{\cA,\cB}f=\natInst\cT{\cA',\cB'}f$
whenever both sides of the equation are defined, \ie{}
$f\in\fin{\cA\limpl\cB}\inter\fin{\cA'\limpl\cB'}$.
Clearly, if $\cT$ has a web, then $\cT$ is type blind iff $\web\cT$ is type blind.
Of course, not all functors in $\Fin$ have a web: 
\begin{counter}
  Let $\cF$ denote the functor of finiteness structures and direct images: 
  $\cF\cA=\pars{\fin\cA,\PFin{\fin\cA}}$ and 
  $\natInst\cF{\cA,\cB}f=\set{(a,\relAppl fa)\st a\in\fin\cA}$.
  The functoriality of $\cF$ is clear as soon as we show 
  that $\natInst\cF{\cA,\cB}f\in\fin{\cF\cA\limpl\cF\cB}$
  when $f\in\fin{\cA\limpl\cB}$. First,
  $\natInst\cF{\cA,\cB}f\subseteq\fin\cA\times\fin\cB$, since $f$ is finitary.
  Moreover, $\natInst\cF{\cA,\cB}f$ is the graph of a function, hence it sends
  $\PFin{\fin\cA}$ to $\PFin{\fin\cB}$. That
  $\relRev{\pars{\natInst\cF{\cA,\cB}f}}$ sends
  $\dual{\fin{\cF\cB}}$ to $\dual{\fin{\cF\cA}}$
  is automatic since $\dual{\fin{\cF\cA}}=\powerset{\fin\cA}$.
\end{counter}

The definition of $\Fin^I$ and of $I$-ary functors in $\Fin$ is
straightforward, and matches exactly that of $\Rel^I$ from $\Rel$.
Order relations on finiteness spaces, together with associated notions of
monotonicity and continuity, will be discussed thoroughly in Section~\ref{section:continuity}.
Let us just remark that $\Fin$ is not cpo-enriched. Indeed,
if $f\subseteq \web\cA\times\web\cB$ is not finitary from $\cA$ to $\cB$, then
$\PFin f$ is a set of finitary relations but it has no finitary upper bound.

On a side note, remark that the construction of a finiteness space by the
transport lemma is not initial, in the sense that the relation $f$ from $A$ to
$\cB$ through which we transport the finiteness structure of $\cB$ is not
finitary from $(A,\transport{\cB,f})$ to $\cB$ in general: although the
condition --- $f$ sends every element of $A$ to a finitary subset of
$\cB$ --- is
necessary for $f$ to be finitary, it is not sufficient.
\begin{counter}
  \label{counter:not:initial}
  The relation $\natInst\supp{\web\cA}$ from $\web{\oc\cA}$ to $\web{\cA}$ is
  not finitary from $\oc\cA$ to $\cA$ whenever $\web{\cA}$ is non-empty:
  let $\alpha\in\web{\cA}$, then
  $\relAppl{\relRev{\natInst\supp{\web\cA}}}\alpha
  \supseteq\prom{\set\alpha}\setminus\set{\emptyMulSet}
  \in\fin{\oc\cA}$ which is an infinite finitary subset, hence 
  $\relAppl{\relRev{\natInst\supp{\web\cA}}}\alpha\not\in\dual[\web{\oc\cA}]{\fin{\oc\cA}}$;
  we conclude by Lemma~\ref{lemma:finitary:relation}.
\end{counter}

The following section explains how functors in $\Fin$ may be derived 
from functors in $\Rel$ \emph{via} the transport lemma.

\subsection{Transport functors}\label{subsec:transport:functor}

Let $I$ be a fixed set of indexes.
Let $T$ be a functor from $\Rel^I$ to $\Rel$. We call \definitive{ownership
relation} on $T$ the data of a quasi-functional lax natural transformation
$\own_i$ from $T$ to the projection functor $\Pi_i$, for all $i\in I$. Notice
that any ownership relation on $T$ satisfies the hypotheses of
Corollary~\ref{coro:transport:family}.    Indeed,   for   any   family
$\fam{\cA}$ of finiteness spaces, we have $\natInst{\own_i}{\fam{\web \cA}} \in
\Rel(T\fam{\web \cA},\web{\cA_i})$ and, since $\natInst{\own_i}{\fam{\web \cA}}$ is
quasi-functional, $\relAppl{\natInst{\own_i}{\fam{\web{\cA}}}}{\ag \alpha}$ is finite for all 
$\ag\alpha\in T\fam{\web{\cA}}$, hence it is finitary in $\cA_i$.
Therefore, $\transport{\fam{\cA},\fam\own}$ is always a
finiteness structure on $T\fam{\web{\cA}}$. We call \definitive{transport
situation} the data of a functor $T$ and an ownership relation $\fam\own$ on
$T$. In such a situation, for all family $\fam{\cA}$ of finiteness spaces, we
write $\lift T{\fam\own}{\fam\cA}$ for the finiteness space
$\pars{T\fam{\web{\cA}},\transport{\fam{\cA},\fam{\own}}}$ and, for all
finitary relation $\fam f$ from $\fam{\cA}$ to $\fam{\cB}$, we write $\lift T{\fam\own}\fam
f=T\fam f$. Notice that $\lift T{\fam\own}$ defines a functor from $\Fin^I$ to $\Fin$ iff $T\fam f$ is
finitary from $\lift T{\fam\own}\fam{\cA}$ to $\lift T{\fam\own}\fam{\cB}$ as soon as $\fam f$ is
finitary from $\fam{\cA}$ to $\fam{\cB}$. In that case, we say $\lift T{\fam\own}$ is the
\definitive{transport functor} deduced from the transport situation $(T,\fam{\own})$.

We now provide sufficient conditions for a transport situation to give rise to
a transport functor. A \definitive{shape relation} on $(T,\fam\own)$ is the data of a
fixed set $S$ of \definitive{shapes} and a quasi-functional lax natural
transformation $\shp$ from $T$ to the constant functor $E_S$ which sends
every set to $S$ and every relation to $\id^S$, subject to the
following additional condition: for all $\ag a\subseteq T\fam A$, if
$\relAppl{\shp}{\ag a}$ is finite and, for all $i\in I$, $\relAppl{\own_i}{\ag
a}$ is finite, then $\ag a$ is itself finite.  

In other words, with every
$T$-element $\ag\alpha\in T\fam A$ is associated a set of shapes
$\relAppl{\shp}{\ag\alpha}$, which is finite (because $\shp$ is
quasi-functional). Moreover shapes are preserved by
$T$-relations; more precisely, if $(\ag\alpha,\ag\beta)\in T\fam f$ then every
shape of $\ag\beta$ is a shape of $\ag\alpha$ (because $\shp$ is a lax natural
transformation). Notice that when $T$ is symmetric, $\relRev{T\fam
f}=T\fam{\relRev f}$, and we actually obtain
$\relAppl{\shp}{\ag\alpha}=\relAppl{\shp}{\ag\beta}$. The additional condition
states that any $T$-subset $\ag a\subseteq T\fam A$ which involves finitely many
shapes and has a finite support in each component is itself
finite.
\begin{lemma}\label{lemma:transport:situation}
  A transport situation on a symmetric functor defines a transport functor as
  soon as it admits a shape relation.
\end{lemma}
\begin{proof}
  Let $(T,\fam\own)$ be a transport situation with $T$ a symmetric functor, 
  and let $\shp$ be a shape relation for this situation.
  By the above discussion on transport situations, we only have to prove
  that $\ag f\eqdef T\fam f$ is a finitary relation from $\lift T{\fam\own}{\fam{\cA}}$
  to $\lift T{\fam\own}{\fam{\cB}}$ as soon as, for all $i\in I$, $f_i$ is a finitary
  relation from $\cA_i$ to $\cB_i$. 
  
  First, let us show that if $\ag a\in\fin{\lift T{\fam\own}{\fam{\cA}}}$, then
  $\relAppl{\ag f}{\ag a}\in
  \fin{\lift T{\fam\own}{\fam{\cB}}} $. Indeed, for all $i\in I$,
  $\relAppl{\own_i}{\relAppl{\ag f}{\ag a}}\subseteq
  \relAppl{f_i}{\relAppl{\own_i}{\ag a}}$
  because $\own_i$ is a lax natural transformation from $T$ to $\Pi_i$.
  Moreover, by the definition of $\fin{\lift T{\fam\own}\fam{\cA}}$,
  $\relAppl{\own_i}{\ag a}\in\fin{\cA_i}$ and then
  $\relAppl{f_i}{\relAppl{\own_i}{\ag a}}\in\fin{\cB_i}$, because $f_i$ is a finitary
  relation.
  
  We are left to prove that for all $\ag\beta\in\web{{\lift T{\fam\own}\fam{\cB}}}$,
  $\ag a'=\relAppl{\relRev{\ag f}}{\ag\beta}\in\fin{\dual{\pars{\lift T{\fam\own}\fam{\cA}}}}$,
  \ie{} for all $\ag a\in\fin{\lift T{\fam\own}\fam \cA}$, $\ag a\inter\ag a'$
  is finite. By the properties of shape relations, it is sufficient to prove
  that $\relAppl{\shp}{\pars{\ag a\inter\ag a'}}$ is
  finite and, for all $i\in I$, $\relAppl{\own_i}\pars{\ag a\inter\ag a'}$ is finite.
  Notice that $T$ being symmetric, we have  $\relRev{\ag f}=T\fam {\relRev f}$.
  Then, since $\shp$ is a lax natural transformation, 
  $\relComp{\shp}{\relRev{\ag f}}\subseteq \shp$.
  We obtain that $\relAppl{\shp}{\pars{\ag a\inter\ag
  a'}}\subseteq\relAppl{\shp}{\ag a'}=\relAppl{\shp}{\relAppl{\relRev{\ag
  f}}{\ag\beta}}\subseteq\relAppl{\shp}{\ag\beta}$ which is finite, 
  since $\shp$ is quasi-functional. Similarly, for all $i\in I$, $\own_i$ is a
  lax natural transformation from $T$ to $\Pi_i$, hence
  $\relComp{\own_i}{\relRev{\ag f}}\subseteq\relComp{\relRev{f_i}}{\own_i}$: we obtain 
  $\relAppl{\own_i}{\ag a'}\subseteq
  \relAppl{\relRev{f_i}}{\relAppl{\own_i}{\ag\beta}}$.
  Since $\own_i$ is quasi-functional $\relAppl{\own_i}{\ag\beta}$ is finite
  and in particular $\relAppl{\own_i}{\ag\beta}\in\fin{\dual{\cB_i}}$: $f_i$
  being a finitary relation, we obtain that
  $\relAppl{\relRev{f_i}}{\relAppl{\own_i}{\ag\beta}}\in\fin{\dual{\cA_i}}$, 
  and thus $\relAppl{\own_i}{\ag a'}\in \fin{\dual{\cA_i}}$.
  By the definition of $\fin{\lift T{\fam\own}\fam \cA}$, we
  also have $\relAppl{\own_i}{\ag a}\in \fin{\cA_i}$, and we conclude that
  $\relAppl{\own_i}{\pars{\ag a\inter
  \ag a'}}\subseteq\pars{\relAppl{\own_i}{\ag a}}\inter\pars{\relAppl{\own_i}{\ag a'}}$ is finite.
\end{proof}
We do not claim the hypotheses of Lemma \ref{lemma:transport:situation} are
minimal. Notice however that the symmetry of $T$ is essential in the proof,
since it allows $\fam\own$ to control the behaviour of $\relRev{T\fam f}$ as
well as of $T\fam f$. Moreover, the existence of a shape relation is crucial,
since some transport situations with symmetric functor do not preserve finitary
relations:
\begin{counter}
  Consider the symmetric functor $S$ of $\nat$-indexed sequences: for all set $A$,
  $SA=A^{\nat}$ and, for all relation $f\subseteq A\times B$, $Sf=\set{(\fam
  \alpha,\fam \beta)\st\forall n\in\nat,\ (\alpha_n,\beta_n)\in f}$. The
  projections $\proj_n=\set{(\fam\alpha,\alpha_n)\st \fam\alpha\in A^\nat}$
  define an ownership relation $\fam\proj$ on $S$.
  Now consider the unique finiteness space $2$ with web $\set{0,1}$.
	Then $S\web{2}=\set{0,1}^{\nat}$ and \[\transport{2,\fam s}=\set{\ag a\st
	\forall n\in\nat,\relAppl{\proj_n}{\ag\alpha}\in\fin{2}}=\powerset{S\web 2};\]
  in particular $\dual{\transport{2,s}}=\PFin{S\web 2}$.
  Now let $f=\set{(0,0),(1,0)}$ which is a finitary relation from $2$ to $2$:
  $Sf$ is not finitary because $\relAppl{\relRev{Sf}}{\family{0}{n\in\nat}}=S\web 2$ 
  which is infinite.
\end{counter}

\begin{example}
  The transport functor $\oc$ in $\Fin$ is derived from the transport situation
  $(\oc,\supp)$, with shape relation $\natInst{\size}{A}=\set{(\ms\alpha,\card\ms\alpha)\st
  \ms\alpha\in\oc A}$.
  The $I$-ary transport functor $\With$ (resp. $\Oplus$) in $\Fin$ is derived
  from the transport situation $\pars{\Oplus,\fam{\rest}}$ (resp. $\pars{\Oplus,\fam{\rest},\indx}$)
  with shape relation $\indx$ (resp. $\emptyset$).
  Finally, we only consider finite tensor products: 
  the binary functor $\tensor$ in $\Fin$ is derived from the transport situation
  $\pars{\tensor,\proj_1,\proj_2}$ with empty shape relation.
  Indeed, infinitary tensor products do not define functors: 
  the functor $S$ in the above counter-example is an instance of
  $\Tensor$ with $I=\nat$.
\end{example}

\section{Continuity and fixpoints}
\label{section:continuity}

In the classical setting of Scott domains and more precisely of
complete partial orders, continuity is the key
property for a function to have a fixpoint, see for
instance~\cite{amadio-curien:domains}. 
It is well known that $\Rel$ endowed with the inclusion order is a complete partial order, 
and that $\Rel$ is cpo-enriched.
Then the continuity of an endofunctor on $\Rel$ boils down to the
commutation of the functor with directed unions of both sets and relations.
Moreover the fixpoint of any $(n+1)$-ary continuous functor exists and is an $n$-ary
functor in $\Rel$.

The situation in $\Fin$ is more complex, if only because the order relations we
consider on finiteness spaces must have something to do with finiteness
structures, whose behaviour w.r.t. the inclusion order on webs is not trivial.
Our first task is thus to describe the different orders derived from set
inclusion that can naturally endow $\Fin$. We put forward two of them: the
largest one, \definitive{finiteness inclusion}, is a cpo on finiteness spaces;
the most restrictive one, \definitive{finiteness extension}, reflects more
closely the inclusion order on webs. We then show the interest of studying both
orders simultaneously: if a sequence $\fam\cA$ of finiteness spaces is
increasing for finiteness extension, then its supremum for finiteness
inclusion is \definitive{exact}, \ie{} it finiteness structure is obtained as the
union of the finiteness structures in the sequence. This property prompts us to 
introduce various notions of continuity for finiteness inclusion, depending on
the exactness of the suprema we consider. We then discuss the continuity of 
transport functors: type blindness is an essential property, in that it
ensures ownership relations are stable under inclusions of webs (Lemma
\ref{lemma:natural:restriction}).
 
Finally, recall $\Fin$ is not cpo-enriched: the least fixpoint of an
$(n+1)$-ary functor in one of its variables might not be functorial in the
others. In the next section, we will however exhibit a restricted class of
$(n+1)$-ary transport functors, the fixpoints of which are $n$-ary functors:
the transport technique is again essential in that development.
At the time of writing, we do not know whether this could be generalized to a
larger class of transport functors.

\subsection{Three order relations on finiteness spaces}

\label{section:continuity:orders}

We can consider two natural orders on finiteness spaces, both based on the
inclusion of webs:
\begin{itemize}
  \item \definitive{finiteness inclusion}: write $\cA \finInc \cB$ if 
    $\web{\cA}\subseteq\web{\cB}$ and $\fin{\cA}\subseteq\fin{\cB}$ ;
  \item \definitive{finiteness extension}: write $\cA \finExt \cB$ if 
    $\web{\cA}\subseteq\web{\cB}$ and $\fin{\cA}=\fin{\cB}\inter\powerset{\web{\cA}}$.
\end{itemize}
Notice  that the  dual construction  is increasing  for  the extension
order: $\cA\finExt\cB$  iff $\dual{\cA}\finExt\dual{\cB}$. In general,
this   does  not   hold  for   finiteness  inclusion:   we   may  have
$\cA\finInc\cB$     and     $\dual{\cA}\not\finInc\dual{\cB}$.    When
$\web{\cA}=\web{\cB}$    we    even    obtain   $\cA\finInc\cB$    iff
$\dual{\cB}\finInc\dual{\cA}$ (whereas, in that case, $\cA\finExt \cB$
iff $\cA=\cB$). Thus we could equivalently consider the order given by
the dual inclusion, $\cA\dualInc\cB$ if $\dual{\cA}\finInc\dual{\cB}$,
in place of $\finInc$.  On  a side note, observe that $\cA \finExt\cB$
iff     we    have     $\cA\finInc    \cB$     and    $\cA\dualInc\cB$
simultaneously. Moreover  $\fsZero$ is the  minimum of each  of these
orders (recall that $\fsZero$ is the empty finiteness space). From
now on,  we consider only  $\finInc$ and $\finExt$: the  properties of
$\dualInc$ are exactly those of $\finInc$ up to finiteness duality.

\begin{lemma}
  Every family $\fam{\cA}$ of finiteness spaces admits a least upper bound
  $\FinSup\fam{\cA}$ (their \definitive{finiteness supremum}) and a
greatest lower bound
  $\FinInf\fam{\cA}$ (their \definitive{finiteness infimum}) for the finiteness
  inclusion order.  They are given by
  $\web{\FinSup{\fam{\cA}}}=\Union\fam{\web{\cA}}$, 
  $\web{\FinInf{\fam{\cA}}}=\Inter\fam{\web{\cA}}$, 
  $\fin{\FinSup{\fam{\cA}}}=\parBidual[\Union\fam{\web{\cA}}]{\Union\fam{\fin{\cA}}}$ and
  $\fin{\FinInf{\fam{\cA}}}=\Inter\fam{\fin{\cA}}$.
  In particular, $\finInc$ is a complete partial order on finiteness spaces.
\end{lemma}
\begin{proof}
  This is a general fact for bidual closure operators.
\end{proof}
In the following, unless otherwise stated, suprema and infima are always
relative to the inclusion order $\finInc$, as described in the previous lemma.

Notice that, in general, $\Union\fam{\fin{\cA}}$ is not a finiteness
structure on $\web{\FinSup\cA}$ by itself, hence the bidual closure in 
$\fin{\FinSup\fam{\cA}}$:
\begin{counter}
  \label{counter:union:not:closed}
  Let $\fA$ be any fake finiteness structure on some set $A$, that is such that $\fA\subsetneq \bidual[A]\fA$.
  For all $f\in\fA$, let $\cA_f=(f,\powerset f)$. Then
  $\Union_{f\in\fA}\fin{\cA_f}=\fA$, but
  $\fin{\FinSup_{f\in\fA}\cA_f}=\bidual[A]\fA$.
\end{counter}
When however $\Union\fam{\fin{\cA}}$  is a  finiteness  structure, we
have: $\fin{\FinSup{\fam\cA}}=\parBidual{\Union\fam{\fin{\cA}}}=\Union\fam{\fin{\cA}}$ 
and we say $\FinSup\fam{\cA}$ is an \definitive{exact} supremum.

Suprema and infima for $\finExt$ do not exist in general, 
even considering the variant up to bijections: $\cA\finExtSim\cB$ if there is 
$\cA'\cong\cA$ such that $\cA'\finExt\cB$.\footnote{This preorder is 
considered by \citet[unpublished preliminary version]{ehrhard:fs} in order
to describe the interpretation of second order quantification of linear logic.} 
\begin{counter}
  Let $\fam{\fF}=\genFam{\fF}n{\nat}$ be the unique sequence of finiteness
  spaces such that, for all $n\in\nat$, $\web{\fF_n}=\zeroTo{n-1}$: then any
  finiteness space of web $\nat$ is a $\finExt$-upper bound of all the
  $\fF_n$'s, hence a $\finExtSim$-upper bound; but, e.g., $\fsFlatNat=(\nat,\PFin{\nat})$
  and $\dual{\fsFlatNat}$ have no common
  $\finExtSim$-lower bound.
\end{counter}
Notice however that in that case $\FinSup\fam\fF=\fsFlatNat$ is an exact supremum.
This remark is actually an instance of a more general fact. 
Indeed:
\begin{lemma}
  \label{lemma:extension:sequence}
  If $\fam \cA$ is an $\finExt$-increasing sequence, then $\FinSup\fam \cA$ is exact.
\end{lemma}
\begin{proof}
  Apply    the    transport    lemma    in    the    form    of    the
  Corollary~\ref{coro:transport:family}                              to
  $A=\Union_{n\in\nat}\web{\cA_n}$   and  to   the   following  $(\set
  *\union\nat)$-indexed families of finiteness spaces and relations:
  \begin{itemize}
  \item $\cB_*=\cN$ and $\forall
n\in\nat$, $\cB_n=\cA_n$
  \item $f_*=\set{(\alpha,n)\st  \alpha\not\in\web{\cA_n}}$  and
   $\forall n\in\nat$, $f_n=\natInst\id{\web{\cA_n}}=\set{(\alpha,\alpha)\st
  \alpha\in\web{\cA_n}}$.
\end{itemize}
Then, using  that $\fam\cA$  is increasing,
  the reader can easily check that:
  \begin{equation*}
    \transport{(\cB_*,\fam\cB),(f_*,\fam f)}=\set{a\subseteq
      \Union\fam{\web{\cA}}\st\exists p\in\nat,\,a\subseteq\web{\cA_p}\land
      \forall n\in\nat,\ a\inter \web{\cA_n}\in\fin{\cA_n}}.
  \end{equation*}
  We conclude that $\Union\fam{\fin{\cA}}=\transport{(\cB_*,\fam\cB),(f_*,\fam f)}$,
  hence $\Union\fam{\fin{\cA}}$ is a finiteness structure.
\end{proof}
Notice that this relies heavily on both the linear ordering of the family and the
extension order as is shown by the following counter-examples.
\begin{counter}[A directed family for finiteness extension]
  Notice that the family of finiteness spaces in
  Counter-example~\ref{counter:union:not:closed}, 
  whose supremum is not exact,
  is however directed for $\finExt$.
\end{counter}
\begin{counter}[An increasing sequence for finiteness inclusion]
  Recall that the sequence of finiteness structures $\genFam{\fC}n{\nat}$ of
  Counter-example~\ref{counter:codaggers} is increasing for inclusion.
  We then form the sequence $\genFam{\cC}n{\nat}$ where 
  $\cC_n=(\nat\times\nat,\fC_n)$, which is increasing for $\finInc$.
  Then $\FinSup\fam\cC$ is not exact, because $\Union\fam\fC$ is a fake finiteness structure.
\end{counter}

Lemma~\ref{lemma:extension:sequence} emphasizes the fact that we should not
focus on finiteness inclusion or finiteness extension separately, but rather
investigate how they can interact. Notice for instance that, as a corollary
of Lemma~\ref{lemma:extension:sequence}, for all $\finExt$-increasing functor
$\cT$ from $\Fin$ to $\Fin$, $\fix \cT=\FinSup_{n\in\nat} \cT^n\fsZero$ is exact.
In the following we show that this defines the least fixpoint of $\cT$ up to some
hypotheses on $\cT$ w.r.t. both finiteness inclusion and finiteness extension.

\subsection{Exact continuity and direct continuity}
\label{section:exact:direct}

A \definitive{directed supremum} is the $\finInc$-supremum of a 
$\finInc$-directed family.
We say a $\finInc$-monotonic functor $\cT$ in $\Fin$ is:
\begin{itemize}
  \item \definitive{weakly continuous} if $\cT$ commutes to directed suprema when they are exact; 
  \item \definitive{exactly continuous} if $\cT$ commutes to exact directed suprema;
  \item \definitive{directly continuous} if $\cT$ commutes to all directed suprema.
\end{itemize}
Let us precise the second case: $\cT$ is exactly continuous iff it is weakly continuous
and $\cT\FinSup\fam{\cA}=\FinSup\fam{\cT\cA}$ is exact for all exact directed supremum
$\FinSup\fam{\cA}$. In particular, both direct continuity and exact
continuity imply weak continuity, but there is no \emph{a priori} implication
between direct continuity and exact continuity: a directly continuous functor may not
preserve exactness; an exactly continuous functor may not preserve non-exact
suprema.  Moreover, notice that exactly continuous (resp. directly continuous)
functors compose, but weakly continuous ones may not: if $\cT$ and $\cU$ are weakly
continous functors and $\FinSup\fam{\cA}$ is an exact directed supremum, we do
not know whether $\FinSup\fam{\cT\cA}$ is exact, hence we can not
deduce that $\cU$ commutes to this supremum.

The main property we shall use about weakly continuous functors (and \emph{a fortiori}
exactly continuous or directly continous ones) is that they admit least fixpoints, 
as soon as they preserve finiteness extensions.
\begin{lemma}\label{lemma:continuity:fixpoint}
  If $\cT$ is a weakly continuous functor, which is moreover $\finExt$-increasing, 
  then $\fix \cT=\FinSup_{n\in\nat} \cT^n\fsZero$ is the ($\finInc$-)least
  fixpoint of $\cT$.
\end{lemma}
\begin{proof}
  We have already remarked in section \ref{section:continuity:orders} that
  $\fix \cT$ is an exact directed supremum: hence $\cT\fix \cT=\FinSup_{n\in\nat}
  \cT^{n+1}\fsZero= \fix \cT$ because $\fsZero$ is minimum.
  Now let $\cY$ be any fixpoint of $\cT$: 
  by iterating the application of $\cT$ to the inequation $\fsZero\finExt\cY$, 
  we obtain $\cT^n\fsZero\finExt \cT^n\cY=\cY$, hence $\cT^n\fsZero\finInc \cY$ 
  for all $n\in\nat$, and finally $\fix \cT\finInc\cY$.
\end{proof}

In order to generalize the definitions of continuity to $I$-ary functors, we adapt 
the same conventions as in the relational setting.
By $\fam{\maf{\cA}}$, we denote an
$I$-indexed family $\genFam{\maf{\cA}}iI$ of families of finiteness spaces,
where each $\maf{\cA}_i=\family{\cA_{i,j}}{j\in J_i}$ takes indices in some
variable set $J_i$. We say $\fam{\maf{\cA}}$ is directed if each
$\maf{\cA}_i$ is directed. We write $\fam{\FinSup\maf{\cA}}$ for
$\family{\FinSup\maf{\cA}_i}{i\in I}=\family{\FinSup_{j\in
J_i}\cA_{i,j}}{i\in I}$ and call this family 
the supremum of $\fam{\maf{\cA}}$: we say this supremum is exact
if each $\FinSup\maf{\cA_i}$ is. Finally, if $\cT$ is a functor from $\Fin^I$ to
$\Fin$, we write $\maf{\cT\fam{\cA}}$ for $\family{\cT\family{\cA_{i,j_i}}{i\in I}}{\fam{j\in
J}}$.

\begin{definition}
Let $\cT$ be a $\finInc$-monotonic functor from $\Fin^I$ to $\Fin$. We say $\cT$ is:
\begin{itemize}
  \item \definitive{exactly continuous} if it commutes to all exact 
    directed suprema, \ie{} $\FinSup\maf{\cT\fam{\cA}}$ is exact 
    and $\cT\pars{\fam{\FinSup\maf{\cA}}}=\FinSup\maf{\cT\fam{\cA}}$ and as soon as 
    $\fam{\maf{\cA}}$ is directed and $\fam{\FinSup\maf{\cA}}$ is exact;
  \item \definitive{directly continuous} if it commutes to all directed suprema,
    \ie{} $\cT\pars{\fam{\FinSup\maf{\cA}}}=\FinSup\maf{\cT\fam{\cA}}$ 
    as soon as $\fam{\maf{\cA}}$ is directed.
\end{itemize}
\end{definition}

The following result follows from the associativity of suprema:
\begin{lemma}\label{lemma:composition:continuity}
  Directly continuous (resp. exactly continuous) functors compose:
  if $\cT$ is a directly continuous (resp. exactly continuous) functor from $\Fin^I$ to
  $\Fin$ and, for all $i\in I$, $U_i$ is a directly continuous (resp. exactly
  continuous) functor from $\Fin^{J_i}$ to $\Fin$ then 
  $\cT\circ \fam U$ is a directly continuous (resp. exactly continuous) functor from 
  $\Fin^{\sum\fam J}$ to $\Fin$.
\end{lemma}
We do not detail the proof as it amounts to a futile exercise in formality: 
we have to consider families of families of families of finiteness spaces, then
simply check that the above definitions apply, up to some juggling with indices.

\subsection{Continuity of transport functors}

\label{section:continuity:transport}

In this section, we establish the properties of type blind transport
functors w.r.t. the order relations on finiteness spaces. 
Notice that the same results would actually hold for arbitrary transport functors, 
provided the properties established in Lemma \ref{lemma:natural:restriction} hold.

\begin{lemma}
  All type blind transport functors are monotonic for both $\finExt$ and
  $\finInc$.
\end{lemma}
\begin{proof}
  Let $\cT$ be a type blind transport functor with ownership relation
  $\fam{\own}$ and assume that $\fam{\cA\finInc\cB}$. Since $\web\cT$ is type blind, 
  Lemma~\ref{lemm:inclusion-preserving-functors} entails that
  $\web{\cT\fam{\cA}}\subseteq\web{\cT\fam{\cB}}$. 
  Then let $\ag a\subseteq\web{\cT\fam{\cA}}$. We have $\ag
  a\in\fin{\cT\fam{\cA}}$ iff for
  all $i\in I$, $\relAppl{\own^{\fam{\web{\cA}}}_i}{\ag a}\in\fin{\cA_i}$.
  Since $\fam{\web{\cA}\subseteq\web{\cB}}$ and $\own_i$ is a lax natural
  transformation, Lemma~\ref{lemma:natural:restriction} implies
  $\relAppl{\own^{\fam{\web{\cA}}}_i}{\ag a}=
  \relAppl{\own^{\fam{\web{\cB}}}_i}{\ag a}$.  Moreover,
  $\fin{\cA_i}\subseteq\fin{\cB_i}$, hence
  $\relAppl{\own^{\fam{\web{\cB}}}_i}{\ag a}\in\fin{\cB_i}$.  We thus obtain
  $\ag a\in\fin{\cT\fam{\cB}}$. We conclude that $\cT\fam{\cA}\finInc\cT\fam{\cB}$.

  If we moreover assume that $\fam{\cA\finExt\cB}$ then 
    $\relAppl{\own^{\fam{\web{\cA}}}_i}{\ag a}\in\fin{\cA_i}$
  iff $\relAppl{\own^{\fam{\web{\cA}}}_i}{\ag a}\in\fin{\cB_i}$ and we obtain 
  $\ag a\in\fin{\cT\fam{\cA}}$ iff $\ag a\in\fin{\cT\fam{\cB}}$.
  We conclude that $\cT\fam{\cA}\finExt\cT\fam{\cB}$.
\end{proof}

\begin{lemma}
  \label{lemma:continuity:transport:exact}
        A type blind transport functor is exactly continuous as soon as its
        underlying web functor is continuous on sets.
\end{lemma}
\begin{proof}
	Let $\cT$ be a type blind transport functor with ownership
	$\fam{\own}$ and assume its web functor $T$ is continous on sets. Let $\fam{\maf\cA}$ be
	directed and such that each $\FinSup \maf{\cA}_i$ is
	exact. We prove that $\cT{\fam{\FinSup\maf\cA}}=\FinSup\maf{\cT{\fam\cA}}$.  First
	notice that the webs $\web{\cT{\fam{\FinSup\maf\cA}}}=T\fam{\Union\maf{\web{\cA}}}$
	and $\web{\FinSup\maf{\cT{\fam\cA}}}=\Union\maf{T\fam{\web{\cA}}}$ are equal
	because $T$ is continuous on sets. We are left to prove the equality of
        finiteness structures, that is $\fin{\cT{\fam{\FinSup\maf\cA}}}=
	\fin{\FinSup\maf{\cT{\fam\cA}}}$ or equivalently 
	$\dual{\fin{\cT{\fam{\FinSup\maf\cA}}}}=
	\dual{\fin{\FinSup\maf{\cT{\fam\cA}}}}$.
        Let's make explicit that by Corollary~\ref{coro:transport:family} and the definition 
        of $\FinSup$:
	\begin{enumerate}[(a) ]
		\item $\ag a'\in\dual{\fin{\cT{\fam{\FinSup\maf\cA}}}}$ iff
			for all $\fam a$ with $a_i\in\fin{\FinSup\maf\cA_i}$ for all $i\in I$, 
			$\ag a'\finPolar\relDiv{\natInst{\fam\own}{\fam{\Union\maf{\web{\cA}}}}}{\fam a}$;
		\item $\ag a'\in\dual{\fin{\FinSup\maf{\cT{\fam\cA}}}}$ iff
			for all $\fam j\in \prod\fam J$ and all 
                        $\ag a\in\fin{\cT{\fam{\cA}_{\fam j}}}$,
			$\ag a'\finPolar\ag a$.
	\end{enumerate}
	We prove both characterizations are equivalent. 
	
	Assume the condition in (a)
	holds and let $\fam j\in \prod\fam J$ and $\ag a\in\fin{\cT{\fam{\cA}_{\fam j}}}$.
  For all $i\in I$, let $a_i=\relAppl{\natInst\own{\fam{\web\cA_{\fam j}}}_i}{\ag
	a}$:
	$a_i\in\fin{\cA_{i,j_i}}\subseteq\Union\maf{\fin{\cA_i}}\subseteq\fin{\FinSup\maf\cA_i}$.
  Then $\ag a\subseteq\relDiv{\natInst{\fam{\own}}{\fam{\web\cA_{\fam j}}}}{\fam a}= 
  \relDiv{\natInst{\fam{\own}}{\fam{\Union\maf{\web\cA}}}}{\fam a}$ by Lemma
  \ref{lemma:natural:restriction}; by condition (a), we deduce that 
  $\ag a'\inter\ag a$ is finite.

	Now assume the condition in (b) holds and let $\fam a$ be such that 
	$a_i\in\fin{\FinSup\maf{\cA}_i}$ for all $i\in I$. Since each of these suprema is exact, \ie{} 
	$\fin{\FinSup\maf{\cA_i}}=\Union\maf{\fin{\cA_{i}}}$, 
	there exists $\fam{j\in J}$ such that $a_i\in\fin{\cA_{i,j_i}}$ for all $i\in I$.
	Hence $\relDiv{\natInst{\fam{\own}}{\fam{\web\cA_{\fam j}}}}{\fam a}
  \in\fin{\cT{\fam{\cA}_{\fam j}}}$ and we conclude by Lemma \ref{lemma:natural:restriction}.

  It remains only to prove that $\FinSup\maf{\cT{\fam\cA}}$ is exact.
  Let $\ag a\subseteq T\fam{\Union\maf{\web{\cA}}}$. We have just proved that 
  $\ag a\in\fin{\FinSup\maf{\cT{\fam\cA}}}$ iff
  $\ag a\in\fin{\cT\fam{\FinSup{\maf\cA}}}$ iff
  for all $i\in I$, $\relAppl{\natInst{\own}{\fam{\Union\maf{\web{\cA}}}}_i}{\ag a}\in
  \fin{\FinSup\maf{\cA_i}}$. Now, because 
  $\FinSup\maf{\cA_i}$ is exact,  $\fin{\FinSup\maf{\cA_i}}=\Union\maf{\fin{\cA_i}}$.
  Thus $\ag a\in\fin{\FinSup\maf{\cT{\fam\cA}}}$ iff for all $i\in I$, 
  there exists $j_i\in J_i$ such that 
  $\relAppl{\natInst{\own}{\fam{\Union\maf{\web{\cA}}}}_i}{\ag a}\in\fin{\cA_{i,j_i}}$.
  Then $\ag a\in\fin{\cT\fam{\cA}_{\fam j}}$, since
  $\relAppl{\natInst{\own}{\fam{\web{\cA_{\fam j}}}}_i}{\ag a}\subseteq
  \relAppl{\natInst{\own}{\fam{\Union\maf{\web{\cA}}}}_i}{\ag a}$ for all $i\in I$
  (again by Lemma \ref{lemma:natural:restriction}).
\end{proof}

\begin{example}
  The web functors of sums, finite multisets and products are all continuous,
  hence $\With$, $\Oplus$, $\oc$ and binary $\tensor$ are exactly continuous.
\end{example}

We say the ownership relation $\fam\own$ is \definitive{local} if, for all
family $\fam A$ and all $i\in I$:
\begin{itemize}
  \item $\relAppl{\own_i}{\pars{\relDiv{\own_i}{a_i}}}=a_i$ 
    for all $a_i\subseteq A_i$;
  \item $\relAppl{\own_{j}}{\pars{\relDiv{\own_i}{a_i}}}=A_j$ 
    for all $a_i\subseteq A_i$ and all $j\not=i$;
  \item $\own_i$ preserves intersections, \ie{}
    $\relAppl{\own_i}{\Inter \maf{\ag a}}=\Inter\maf{\relAppl{\own_i}{\ag a}}$
    for all $\maf{\ag a\in T\fam A}$.
\end{itemize}
Intuitively, an ownership relation is local if its components do not interact with each other.
In particular, if $\fam\own$ is local then
$\relAppl{\own_i}{\pars{\relDiv{\fam\own}{\fam a}}}=a_i$ for all $i\in I$.

\begin{lemma}
  \label{lemma:continuity:transport:direct}
	A type blind transport functor is directly continuous as soon as its underlying web
	functor is continuous and its ownership relation is local.
\end{lemma}
\begin{proof}
	The proof differs from the previous one only in the direction 
	(b) to (a), where we used the exactness condition, which is no longer
	available. So, assuming (b) and in order to establish (a),
  we first prove the following intermediate result:
  $\relAppl{\natInst{\own}{\fam{\Union\maf{\web{\cA}}}}_i}{\ag
	a'}\in\dual{\fin{\FinSup\maf{\cA}_i}}=\parDual{\Union\maf{\fin{\cA_i}}}$ for all
  $i\in I$. Indeed, let $\fam j\in\prod\fam J$ and $\fam{a\in\fin{\cA_{\fam j}}}$
  (in particular we chose $j_i$ to be any index in $J_i$ and 
  $a_i$ to be any finitary subset of $\cA_{i,j_i}$) and write 
  $\ag a''=\relDiv{\natInst{\fam\own}{\fam{\Union\maf{\web{\cA}}}}}{\fam a}$:
  by Lemma \ref{lemma:natural:restriction},
  $\ag a''=\relDiv{\natInst{\fam\own}{\fam{\web{\cA_{\fam j}}}}}{\fam a}
  \in\fin{\cT{\fam{\cA}_{\fam j}}}$ and thus 
  $\ag a=\ag a'\inter\ag a''$ is finite.  
  Moreover, for all $i\in I$, 
  $\relAppl{\natInst{\own}{\fam{\Union\maf{\web{\cA}}}}_i}{\ag a''}=a_i$, 
  because $\fam\own$ is local. 
  Hence  $\pars{\relAppl{\natInst{\own}{\fam{\Union\maf{\web{\cA}}}}_i}{\ag a'}}\inter a_i
  =\pars{\relAppl{\natInst{\own}{\fam{\Union\maf{\web{\cA}}}}_i}{\ag a'}}\inter
  \pars{\relAppl{\natInst{\own}{\fam{\Union\maf{\web{\cA}}}}_i}{\ag a''}}=
  \relAppl{\natInst{\own}{\fam{\Union\maf{\web{\cA}}}}_i}{\ag a}$ 
  because $\own_i$ preserves intersections.
  Since $\ag a$ is finite and $\own_i$ is quasi-functional, we conclude that
  $\relAppl{\natInst{\own}{\fam{\Union\maf{\web{\cA}}}}_i}{\ag a'}\finPolar a_i$. 
  Since this holds for all
  $a_i\in\Union\maf{\fin{\cA_i}}$, we obtain 
  $\relAppl{\natInst{\own}{\fam{\Union\maf{\web{\cA}}}}_i}{\ag a'}\in
	\dual{\fin{\FinSup\maf{\cA}_i}}$.

	Then let $\fam a$ be such that $a_i\in\fin{\FinSup\maf\cA_i}$ for all $i\in I$ and write 
  $\ag a''=\relDiv{\natInst{\fam\own}{\fam{\Union\maf{\web{\cA}}}}}{\fam a}$:
	we must show that $\ag a'''=\ag a'\inter\ag a''$
	is finite. For all $i\in I$, $a'''_i
  =\relAppl{\natInst{\own}{\fam{\Union\maf{\web{\cA}}}}_i}{\ag a'''}\subseteq 
	\pars{\relAppl{\natInst{\own}{\fam{\Union\maf{\web{\cA}}}}_i}{\ag a'}}\inter a_i$
  is finite: hence 
	$a'''_i\in\fin{\cA_{i,j}}$ for any $j\in J_i$ such that 
	$a'''_i\subseteq\web{\cA_{i,j}}$. Fix $j_i$ to be one such $j$ for all $i\in I$.
	We obtain $\ag a'''\in\fin{\cT{\fam{\cA}_{\fam j}}}$. 
	We conclude since $\ag a'\finPolar\ag a'''$ and thus 
	$\ag a'\inter \ag a'''=\ag a'''$ is finite.
\end{proof}

\begin{example}
  Since $\fam{\rest}$, $\indx$ and $\supp$ are local, 
  $\With$, $\Oplus$ and $\oc$ are directly continuous.
\end{example}

The conditions under which we proved direct continuity of transport 
functors are not minimal. For instance $\fam{\proj}$ is not local
even for $I=\set{1,2}$: since $A\times\emptyset=\emptyset$, 
$\relDiv{(\natInst\proj{A,\emptyset}_1,\natInst\proj{A,\emptyset}_2)}(a,\emptyset)=\emptyset$ 
for all $a\subseteq A$ and then
$\relAppl{\natInst\proj{A,\emptyset}_1}{\emptyset}=\emptyset\not=a$ in general.
However:
\begin{lemma}
  \label{lemma:continuity:transport:tensor}
	Finite tensor products are directly continuous.
\end{lemma}
\begin{proof}
	It is sufficient to consider binary tensor products and prove continuity 
  w.r.t. one of the parameters.
  Let $\maf{\cB}=\genFam{\cB}jJ$ be a directed supremum of finiteness spaces:
  we prove $\cA\tensor\FinSup\maf\cB=\FinSup\family{\cA\tensor\cB_j}{j\in J}$
  or, equivalently, 
  $\parDual{\cA\tensor\FinSup\maf\cB}=\parDual{\FinSup\family{\cA\tensor\cB_j}{j\in J}}$.
  That $\dual{\fin{\cA\tensor\FinSup\maf\cB}}\subseteq
  \dual{\fin{\FinSup\family{\cA\tensor\cB_j}{j\in J}}}$ goes by the same argument 
  as in Lemma~\ref{lemma:continuity:transport:exact}.

  Assume that $c'\in\dual{\fin{\FinSup\family{\cA\tensor\cB_j}{j\in J}}}$:
  we prove $c'\in\dual{\fin{\cA\tensor\FinSup\maf\cB}}
  =\fin{\cA\limpl\dual{\pars{\FinSup\maf\cB}}}$.
  If $a\in \fin{\cA}$ then $\relAppl{c'}{a}\in\dual{\fin{\FinSup\maf\cB}}$.
  Indeed, for all $j\in J$ and $b\in\fin{\cB_j}$, 
  we have $a\times b\in\fin{\cA\tensor\cB_j}\subseteq\fin{\FinSup\family{\cA\tensor\cB_j}{j\in J}}$, 
  hence $c'\finPolar a\times b$: then $\relAppl{c'}{a}\finPolar b$.
  In the other direction, let $\beta\in\web{\FinSup{\maf{\cB}}}=\Union\maf{\web{\cB}}$:
  let $j\in J$ such that $\beta\in\web{\cB_j}$. Then, for all $a\in\fin{\cA}$,
  $a\times\set\beta\in\fin{\cA\tensor\cB_j}$, hence $c'\finPolar a\times\set\beta$
  and we obtain $\relAppl{\relRev{c'}}\beta\finPolar a$.
  We have thus proved that $\relAppl{\relRev{c'}}\beta\in\fin{\dual{\cA}}$, 
  which concludes the proof.
\end{proof}
It is still unclear to us if this argument can be adapted to lift the condition 
on the locality of $\fam\own$ in Lemma~\ref{lemma:continuity:transport:direct},
and thus generalize direct continuity to all type blind transport functors with
continuous web functors.

To sum up, remind from Section~\ref{section:exact:direct} that direct continuity
and exact continuity imply weak continuity, from
Lemma~\ref{lemma:continuity:fixpoint} that weakly continuous functors admit
fixpoints and finally from Lemma~\ref{lemma:composition:continuity}
that direct continuity and exact continuity are stable under
composition.
Then, we can infer that the functors (in one variable) resulting of the
composition of $\With$, $\Oplus$, $\oc$ and finite $\tensor$ admit fixpoints. 

The status of the linear arrow functor $\limpl$ is still unclear in that regard,
if only because its properties with relation to the orders on finiteness spaces
are not straightforward. First notice that
$\cA\limpl\cB=\parDual{\cA\tensor\dual\cB}$ is functorial in $\cB$ but cofunctorial
in $\cA$. By contrast, $\cA\limpl\cB$ is $\finExt$-increasing in both $\cA$
and $\cB$, whereas it is only $\finInc$-increasing in $\cB$ and is not
monotonic in $\cA$.

\section{The finitary relational model of the \texorpdfstring{$\lambda$}{lambda}-calculus}

\label{section:lambda}

It is a well known fact that $\Rel$ is a model of classical linear logic 
\citep[Appendix]{ehrhard:fs}, and even of differential linear logic where:
\begin{itemize}
  \item  linear negation is the transpose cofunctor;
  \item  multiplicatives are interpreted by cartesian products;
  \item  additives are interpreted by disjoint unions;
  \item  exponentials are interpreted by finite multisets.
\end{itemize}
The category of finiteness spaces and finitary relations $\Fin$ is also a model
of classical linear logic, which is the subject of the first part of Ehrhard's
paper \citeyearpar{ehrhard:fs}. This result could actually be
stated as follows: for any finiteness structure we impose on the relational
interpretation of atomic formulas, the relational semantics of a proof is
always finitary in the finiteness space denoted by its conclusion. In other
words, that $\Fin$ is a model of linear logic can be stated as a
\emph{property} of the interpretation of linear logic in $\Rel$. This viewpoint
fits very well with the previous developments of our paper, in which we explore
how distinctive constructions and properties of $\Rel$ can be transported to $\Fin$.
In the present section, we extend this stand to the study of the $\lambda$-calculus, 
which will allow us to discuss datatypes in the next section.

From  the  relational  model  of   linear  logic,  we  can  derive  an
extensional  model  of  the  simply typed  $\lambda$-calculus  by  the
co-Kleisli  construction:  this  gives  rise  to  a  cartesian  closed
category $\RelCC$. Objects in $\RelCC$ are sets and morphisms from $A$
to $B$ are \definitive{multirelations}, that is subsets of $A \CCarrow
B=\MFin A\times B$. Notice that this definition is an instance of Girard's 
translation of the intuitionistic arrow: $A\CCarrow B=\oc  A\limpl B$.
Composition of multirelations is given by
\[g\CCcomp f=\set{\pars{\sum_{i=1}^n
\ms\alpha_i,\gamma}\st n\in\nat\land\exists\ms\beta=\fullMulSet\beta  n\in \oc B,\,
(\ms\beta,\gamma)\in g\land \forall i\,(\ms\alpha_i,\beta_i)\in f}\]
as soon as $f\in\RelCC(A,B)$ and $g\in\RelCC(B,C)$.
The identity multirelation on $A$ is the \definitive{dereliction}: 
$\natInst\dere A=\set{(\mulSet \alpha,\alpha)\st \alpha\in A}$.
The cartesian product is given by the disjoint union of sets 
$\With\fam A$, with projections $\fam{\CCproj}=\fam{\relComp{\rest}{\dere}}$.
If, for all $i\in I$, $f_i\in\RelCC(A,B_i)$, then the unique morphism $\tuple{\fam f}$ 
from $A$ to $\With\fam B$ such that $\proj_i\CCcomp \tuple{\fam f}=f_i$ for all 
$i$ is $\set{\pars{\ms\alpha,(i,\beta)}\st(\ms\alpha,\beta)\in f_i,\,i\in I}$.
The terminal object denoted $\top$ is the empty set $\emptyset$, 
the unique multirelation from $A$ to $\emptyset$ being empty.
The adjunction for closedness is $\RelCC(A\with B,C)\cong\RelCC(A,\oc B\limpl C)$ 
which boils down to the natural bijection $\oc(A\with B)\cong\oc A\tensor\oc B$.

\subsection{Relational interpretation and finiteness property}
In this section, we give an explicit description of the interpretation in
$\RelCC$ of the basic constructions of simply typed $\lambda$-calculi with products.  
Type and term
expressions are given by: \[ A,B \eqdef X \mid A\CCarrow B \mid A\with B \mid
\top \quad\quad\textrm{and}\quad\quad s,t \eqdef
x \mid a \mid \labs xs \mid \appl st \mid \pair st \mid \lpi s \mid \rpi s \mid
\unit\] where $X$ ranges over a fixed set 
$\atoms$ of atomic types, $x$ ranges over term variables and $a$ ranges over term constants.
Of course, the variable $x$ is bound by the abstraction in $\labs xs$, we
consider terms up to $\alpha$-equivalence, and we denote by $\subst sxt$ the
capture-avoiding substitution of term $t$ for variable $x$ in $s$.

To each variable or constant, we associate a type, so that each type admits
infinitely many variables.\footnote{The type system we present is
thus in the style of Church rather than in the style of Curry: typing is syntax
directed. This is only a technical convenience and the remaining of
the paper could very well be recast in a Curry-style setting.}
We write $\constants_A$ for the collection of constants of type $A$.
A typing judgement is an expression $\Gamma\jug s:A$ derived from the rules in
Figure \ref{fig:lc} where contexts $\Gamma$ and $\Delta$ range over finite
lists $(x_1:A_1,\dotsc,x_n:A_n)$ of typed variables.
Since we do not impose Barendregt's convention on λ-terms, variables occuring
in a context need not be pairwise distinct, hence the shape of rule (Var).
If a term $s$ is typable, then its type is uniquely determined, say $A$,
and then $\Gamma\jug s:A$ iff $\Gamma$ contains the free variables of $s$.
The operational semantics of a typed $\lambda$-calculus is given by a
contextual equivalence relation $\converts$ on typed terms:
if $s\converts t$, then $s$ and $t$ have the same type, say $A$;
we then write $\Gamma\jug s\converts t:A$ for any suitable $\Gamma$.
We write $\converts_0$ for the least one such that 
$\lpi{\pair st} \converts_0 s$, $\rpi{\pair st} \converts_0 t$ and $\appl{(\labs xs)}t
\converts_0 \subst sxt$ (provided $t$ and $x$ have the same type).

\begin{figure}[t]
	\begin{center}
		\AxiomC{$x\not\in\Delta$}
		\LBL{Var}
		\UnaryInfC{$\Gamma,x:A,\Delta\fCenter x:A$}
		\DisplayProof
 		\quad\quad
 		\AxiomC{}
  		\LBL{Unit}
  		\UnaryInf$\Gamma\fCenter \unit:\top$
  		\DisplayProof
 		\quad\quad
 		\AxiomC{$a \in \constants_A$}
		\LBL{Const}
		\UnaryInfC{$\Gamma\fCenter a:A$}
		\DisplayProof
		\smallskip
		\\
		\AxiomC{$\Gamma,x:A\fCenter s:B$}
		\LBL{Abs}
		\UnaryInfC{$\Gamma \fCenter \labs x{s}:A\CCarrow B$}
		\DisplayProof
		\quad
		\quad
		\AxiomC{$\Gamma\fCenter {s}:A\CCarrow B$}
		\AxiomC{$\Gamma\fCenter {t}:A$}
		\LBL{App}
		\BinaryInfC{$\Gamma\fCenter \appl {s}{t}:B$}
		\DisplayProof
		\smallskip
		\\
		\AxiomC{$\Gamma\fCenter {s}:A$}
		\AxiomC{$\Gamma\fCenter {t}:B$}
		\LBL{Pair}
		\BinaryInfC{$\Gamma\fCenter \pair {s}{t}:A \with B$}
		\DisplayProof
		\quad
		\AxiomC{$\Gamma\fCenter {s}:A\with B$}
		\LBL{Left}
		\UnaryInfC{$\Gamma \fCenter \lpi s:A$}
		\DisplayProof
		\quad
		\AxiomC{$\Gamma\fCenter {s}:A\with B$}
		\LBL{Right}
		\UnaryInfC{$\Gamma \fCenter \rpi s:B$}
		\DisplayProof
	\end{center}
	\caption{Rules of typed $\lambda$-calculi with products
	\label{fig:lc}}
\end{figure}

\begin{figure}[t]
	\begin{center}
		\AxiomC{$x\not\in\Delta$}
		\SLBL{Var}
		\UnaryInfC{
		$\Gamma^{\emptyMulSet},
		x^{\mulSet\alpha}:A,
		\Delta^{\emptyMulSet}
		\fCenter x^{\alpha}:A$}
		\DisplayProof
   \quad
    \quad
    \quad
		\AxiomC{$a\in\constants_A$}
		\AxiomC{$\alpha \in \sem a$}
		\SLBL{Const}
		\insertBetweenHyps{\ }
		\BinaryInfC{
		$\Gamma^{\emptyMulSet}
		\fCenter a^\alpha:A$}
		\DisplayProof
		\smallskip

		\AxiomC{$\Gamma,x^{\ms \alpha}:A \fCenter s^\beta:B$}
		\SLBL{Abs}
		\UnaryInfC{$\Gamma \fCenter \labs x{s}^{(\ms\alpha,\beta)}:A
		\CCarrow B$}
		\DisplayProof
		\hfill
		\AxiomC{$\Gamma_0\fCenter {s}^{(\fullMulSet\alpha
		k,\beta)}:A\CCarrow B
		\quad
		\quad
		\Gamma_1\fCenter {t}^{\alpha_1}:A
		\quad\cdots\quad
		\Gamma_k\fCenter {t}^{\alpha_k}:A
		$}
		\SLBL{App}
		\UnaryInfC{$\sum_{j=0}^k\Gamma_j\fCenter\appl st^\beta:B$}
		\DisplayProof
		\smallskip

		\AxiomC{$\Gamma\fCenter s_i^\alpha:A_i$}
		\SLBL{Pair$_i$}
		\UnaryInfC{$\Gamma\fCenter\pair {s_1}{s_2}^{(i,\alpha)}:A_1\with A_2$}
		\DisplayProof
		\quad
		\quad
		\AxiomC{$\Gamma\fCenter s^{(1,\alpha)}:A\with B$}
		\SLBL{Left}
		\UnaryInfC{$\Gamma \fCenter \lpi s^\alpha:A$}
		\DisplayProof
		\quad
		\quad
		\AxiomC{$\Gamma\fCenter s^{(2,\beta)}:A\with B$}
		\SLBL{Right}
		\UnaryInfC{$\Gamma \fCenter \rpi s^\beta:B$}
		\DisplayProof
	\end{center}
	\caption{Computing points in the relational semantics
	\label{fig:sem}}
\end{figure}

Assume a set $\sem X$ is given for each atomic type $X$; then
we interpret type constructions by $\sem{A\CCarrow B} = \sem A\CCarrow\sem B$,
$\sem{A\with B} = \sem A \with \sem B$ and $\sem \top = \emptyset$. Further assume
that with every constant $a\in\constants_A$ is associated a subset
$\sem a\subseteq \sem A$. The relational semantics of a derivable typing
judgement $x_1:A_1,\dotsc,x_n:A_n\jug s:A$ will be an $n$-ary multirelation
$\sem{s}_{x_1:A_1,\dotsc,x_n:A_n}
\subseteq \sem{A_1\CCarrow\cdots\CCarrow A_n\CCarrow A}$.
We first introduce the deductive system of Figure \ref{fig:sem}, which is 
a straightforward adaptation of \citeauthor{carvalho:exec}'s system $R$
\citeyearpar{carvalho:exec} to the simply typed case. In this system,
derivable judgements are semantic annotations of typing judgements: 
$x_1^{\ms\alpha_1}:A_1,\dotsc,x_n^{\ms\alpha_n}:A_n\jug s^\alpha:A$
stands for
$(\ms\alpha_1,\dotsc,\ms\alpha_n,\alpha)\in\sem{s}_{x_1:A_1,\dotsc,x_n:A_n}$
where each $\ms\alpha_i\in\MFin{\sem{A_i}}$
and $\alpha\in\sem A$.
In rules $\sem{\text{Var}}$ and $\sem{\text{Const}}$,
$\Gamma^{\emptyMulSet}$ denotes an annotated context of the form
$x_1^{\emptyMulSet}:A_1,\dotsc,x_n^{\emptyMulSet}:A_n$.
In rule $\sem{\text{App}}$, the sum of annotated contexts is defined 
pointwise:
$\pars{x_1^{\ms\alpha_1}:A_1,\dotsc,x_n^{\ms\alpha_n}:A_n}+
\pars{x_1^{\ms\alpha'_1}:A_1,\dotsc,x_n^{\ms\alpha'_n}:A_n}=
\pars{x_1^{\ms\alpha_1+\ms\alpha'_1}:A_1,\dotsc,x_n^{\ms\alpha_n+\ms\alpha'_n}:A_n}$.
The semantics of a term is given by:
$\sem{s}_{x_1:A_1,\dotsc,x_n:A_n}=\set{
(\ms\alpha_1,\dotsc,\ms\alpha_n,\alpha)\st
x_1^{\ms\alpha_1}:A_1,\dotsc,x_n^{\ms\alpha_n}:A_n\jug s^\alpha:A }$.
Notice that there is no rule for $\unit$ in Figure \ref{fig:sem}, because
$\sem{\unit}_\Gamma=\emptyset$.
Since we follow the standard
interpretation of typed $\lambda$-calculi in cartesian closed
categories (see ~\citet{lambek-scott:hocl}), in the particular case of $\RelCC$, we obtain:
\begin{lemma}[Invariance]
	\label{lemma:invariance}
	If $\Gamma\jug s\converts_0 t :A$ then $\sem s_\Gamma = \sem t_\Gamma$.
\end{lemma}

For all finiteness spaces $\cA$ and $\cB$, write $\cA\CCarrow
\cB=\oc\cA\limpl\cB$. 
\begin{lemma}
  \label{lemma:ccarrow}
  Let $f$ be a multirelation from $\web{\cA}$ to $\web{\cB}$. Then 
  $f\in\fin{\cA\CCarrow\cB}$ iff, for all $a\in\fin{\cA}$, 
  $\relAppl f{\prom a}\in\fin\cB$ and for all $\beta\in\web{\cB}$, 
  $\relAppl{\relRev f}\beta\finPolar\prom a$.
\end{lemma}
\begin{proof}
By Lemma~\ref{lemma:finitary:relation}, $f\in\fin{\cA\CCarrow\cB}$ iff 
for all $\ms a\in\fin{\oc\cA}$, 
  $\relAppl f{\ms a}\in\fin\cB$ and for all $\beta\in\web{\cB}$,
  $\relAppl{\relRev f}\beta\finPolar \ms a$.
By the characterization of  $\oc\cA$  given in Example~\ref{ex:oc}: 
\[\fin{\oc\cA}
=\set{\ms a\subseteq\oc{\web\cA}\st\relAppl{\natInst{\supp}{\web{\cA}}}{\ms a}\in\fin\cA}
=\bidual[\web\cA]{\set{\prom{a}\st{}a\in\fin{\cA}}}.\] 
Then the result follows from the inclusions $\set{\prom
a\st a\in\fin\cA}\subseteq\fin{\oc\cA}$ and 
$a\subseteq \support{\prom a}$ for all $a\subseteq\web\cA$.
\end{proof}

We call  \definitive{finitary multirelations} from $\cA$  to $\cB$ the
elements  of  $\fin{\cA\CCarrow\cB}$. Then  the  category $\FinCC$  of
finiteness  spaces and finitary  multirelations is  no other  than the
co-Kleisli   category    derived   from   $\Fin$.     The   relational
interpretation  $\sem{\cdot}$  of simply typed  $\lambda$-calculi thus
defines  a semantics  in  $\FinCC$ as  follows.   Assume a  finiteness
structure $\fin X$  on $\sem X$ is given for all  atomic type $X$, and
write $\fsem  X$ for the  finiteness space $\pars{\sem X,\fin  X}$. We
set $\parFsem{A\CCarrow B}=\fsem  A\CCarrow \fsem B$, $\parFsem{A\with
  B}=\fsem  A\bipr \fsem B$  and $\fsem\top=\fsEmpty$.   Then, further
assuming that,  for all $a\in\constants_A$,  $\sem a\in\fin{\fsem A}$,
we obtain:
\begin{lemma}[Finiteness]
  \label{lemma:finiteness}
  If $x_1:A_1,\dotsc,x_n:A_n\jug s:A$ then
	\[\sem{s}_{x_1:A_1,\dotsc,x_n:A_n}
	\in\fin{\fsem{A_1}\impl\cdots\impl\fsem{A_n}\impl\fsem A}.\]
\end{lemma}

\subsection{On the relations denoted by \texorpdfstring{$\lambda$}{lambda}-terms}
Pure typed $\lambda$-calculi are
those with no additional constant or conversion rule: fix a set $\atoms$ of
atomic types, and write $\Lambda^{\atoms}_0$ for the calculus where
$\constants_{A}=\emptyset$ for every type $A$, and $s \converts t$ iff $s \converts_0 t$.
This is the most basic case and we have just shown that $\RelCC$ and $\FinCC$ model 
$\converts_0$. Be aware that if we introduce no atomic type, then the semantics
is actually trivial: in $\Lambda^{\emptyset}_0$, all types and terms are interpreted by
$\emptyset$.

By contrast, we can consider the internal language $\LambdaRel$ of $\RelCC$ in
which all relations can be described as terms: fix the atomic types $\atoms$ as the collection
of all sets and the constants $\constants_{A}=\powerset{\sem{A}}$. Then set
$s\converts_{\Rel} t$ iff $\sem s_\Gamma = \sem t_\Gamma$, for any suitable $\Gamma$.
The point in defining such a monstrous language is to enable very natural
notations for relations: in general, we will identify closed terms in
$\LambdaRel$ with the relations they denote in the empty context. For instance,
we write $\natInst\dere A=\labs xx$ with $x$ of type $A$; and 
if $f\in\RelCC(A,B)$ and $g\in\RelCC(B,C)$, we have 
$g\CCcomp f = \parLabs x{\parAppl g {\appl fx}}$. 
More generally, if $s$ and $t$ are terms in $\LambdaRel$ of type $A$ in context 
$\Gamma$, we may simply write $\Gamma\jug s=t:A$ for
$\sem{s}_{\Gamma}=\sem{t}_{\Gamma}\in\sem A$. 
Similarly, 
the internal language $\LambdaFin$ of $\FinCC$,
where $\atoms$ is the collection of all finiteness spaces and 
$\constants_{A}=\fin{\fsem{A}}$, allows to denote conveniently all finitary relations
and equations between them.

Before we address the problem of algebraic types, we review some basic
properties of the semantics. First, $\RelCC$ and $\FinCC$ being cartesian
closed categories, they actually model typed
$\lambda$-calculi with extensionality:
$s:A\CCarrow A\jug\labs x{\pars{sx}}=s:A\CCarrow A$
as soon as $x$ is not free in $s$.
Moreover, they admit all products, and they are models of 
$\lambda$-calculi with surjective tuples of arbitrary arity, that is
$t:\With\fam
A\jug\tuple{\CCproj_i t}_{i\in I}=t$.  In accordance with this last remark, we
may identify any variable of type $\With\fam A$ with a tuple
and write, e.g., $\CCproj_i=\labs{\fam x}{x_i}$.

Being cpo-enriched, $\RelCC$ admits fixpoints at all types and the least fix
point operator on $A$ is the least multirelation $\fp\subseteq (A\CCarrow
A)\CCarrow A$ such that $\fp=\parLabs f{\appl f{\pars{\appl{\fp} f}}}$, \ie{}
$\fp=\Union_{n\in\nat}\fp_n$ where $\fp_0=\emptyset$ and $\fp_{n+1}=\parLabs
f{\appl f{\pars{\appl{\fp_n} f}}}$ (see for
instance~\citet[Chapter 6]{amadio-curien:domains}). More explicitly, we get:
\[\fp_{n+1}=\set{
	\bigg(\mulSet{\pars{\fullMulSet\alpha p,\alpha}}+\sum_{k=1}^p\ms\phi_k,\alpha\bigg)
\st p\in\nat\land\forall k\in\oneTo p,\,\pars{\ms\phi_k,\alpha_k}\in\fp_n}.\]
Notice that, for all $n\in\nat$ and all finiteness space $\cA$,
$\natInst{\fp}{\web{\cA}}_n\in\fin{(\cA\CCarrow\cA)\CCarrow\cA}$.  But in
general, $\fp$ is not finitary: \citet[Section 3]{ehrhard:fs} details a
counter-example, but we can actually show that the least fixpoint operator is
never finitary on non-empty webs (and thus no fixpoint operator is, since
finiteness structures are downward closed for inclusion).
\begin{lemma}
  If $\web\cA\not=\emptyset$, then
  $\natInst{\fp}{\web{\cA}}\not\in\fin{(\cA\CCarrow\cA)\CCarrow\cA}$.
\end{lemma}
\begin{proof}
  Let $\alpha\in\web\cA$ and
  $f=\set{(\emptyMulSet,\alpha),(\mulSet\alpha,\alpha)}
  \in\PFin{\cA\CCarrow\cA}\subseteq\fin{\cA\CCarrow\cA}$.  Observe that
  $(\mulSet{(\emptyMulSet,\alpha)},\alpha)\in\fp_1$,
  $(\mulSet{(\emptyMulSet,\alpha),(\mulSet\alpha,\alpha)},\alpha)\in\fp_2$, and
  more generally
  $(\mulSet{(\emptyMulSet,\alpha)}+n\mulSet{(\mulSet\alpha,\alpha)},\alpha)\in\fp_{n+1}$.
  Hence $\prom f\not\finPolar\relAppl{\relRev{\fp}}\alpha$ and we conclude 
  by Lemma \ref{lemma:ccarrow}.
\end{proof} 

This result indicates that the finitary semantics refuses infinite computations
and will not accomodate general recursion, \eg{} in the sense of PCF. It is
thus very natural to investigate the nature of the
algorithms that can  be studied in a finitary  setting. It was already
known from Ehrhard's original paper \citeyearpar{ehrhard:fs} that one can
model a  restricted form of  tail-recursive iteration. In  recent work
\citep{vaux:relT},  the  second  author   showed  that   the  finitary
relational model of the $\lambda$-calculus can actually be extended to
Gödel's system  $T$, \ie{} typed recursion on  integers. The remaining
of the  paper provides  a generalization of  this result  to recursive
algebraic datatypes.

\section{Lazy recursive algebraic datatypes}

\label{section:rec}

An algebraic datatype is a composite of products, sums and base
types: products are equipped with projections and a tupling operation (\ie{}
pairing, in the binary case), while sums are equipped with injections and a
case definition operator (which is essentially 
a pattern matching operator). Of course, datatype constructors are meant to be
polymorphic: in other words they are particular functors. In a cartesian closed
category, it is only natural to interpret products as categorical products. On
the other hand, coproducts are not always available, hence the interpretation
of sums might not be as canonical.

In this concluding section of our paper, we first discuss the status of sums
in $\RelCC$ and $\FinCC$. We are then led to investigate the semantics of
recursive algebraic datatypes we obtain by taking the fixpoints of algebraic
functors. In particular, we remark that the relational
interpretation gives rise to a \emph{lazy} semantics. For instance the web of
the datatype of trees is not a set of trees but a set of paths in trees: this
generalizes a similar feature of the coherence semantics of system $T$
\citep{girard:prot} and its relational variant \citep{vaux:relT}.  We finish the
paper by providing an explicit description of the relational interpretation of
the constructors and destructors of recursive algebraic datatypes, which
enables us to prove them finitary.

\subsection{Sums}

By contrast with the cartesian structure, the cocartesian structure is ruled
out by the co-Kleisli construction from $\Rel$ to $\RelCC$ (as by the one from
$\Fin$ to $\FinCC$): $\RelCC$ does not have coproducts. 

\begin{counter}
  There is no coproduct for the pair of sets $(\emptyset,\emptyset)$ in
  $\RelCC$. Indeed, assume that there exists a set $A$ and multirelations $i_0$ and
  $i_1$ from $\emptyset$ to $A$, such that for all set $B$ and all
  multirelations $f_0$ and $f_1$ from $\emptyset$ to $B$ there exists a unique
  $h\in\RelCC(A,B)$ such that $h\CCcomp i_k=f_k$  for $k=0,1$.  Necessarily,
  there exist $\alpha_0$ and $\alpha_1$, such that $(\emptyMulSet,\alpha_k)\in
  i_k$ but $(\emptyMulSet,\alpha_k)\not\in i_{1-k}$ for $k=0,1$: otherwise, \eg,
  $h\CCcomp i_0\subseteq h\CCcomp i_1$ for all $h$. Now consider relations from
  $A$ to $\{0\}$,
  $h'=\set{(\mulSet{\alpha_0},0),(\mulSet{\alpha_1},0)}$ and
  $h''=\set{(\emptyMulSet,0)}$: we have $h'\CCcomp i_k=h''\CCcomp i_k$ for
  $k=0,1$ but $h'\not=h''$, which contradicts the unicity property of the coproduct.
\end{counter}

We can however provide an adequate interpretation of sum types, adapting
Girard's interpretation of intuitionistic logic in coherence spaces
\citep{girard:prot}.  We write $A\loplus B$ for the lifted sum $\set{1,2}\union A\oplus B$ 
of $A$ and $B$, and more generally: $\Loplus\fam A=I\union\Oplus\fam
A$.\footnote{
Another possibility for interpreting sums is to consider 
$A\ocplus B=\oc A\oplus\oc B$ which is preferred by Girard to interpret
intuitionistic disjunction because it enjoys an extensionality property.
There is no doubt we could adapt the following sections of our paper 
to this notion of sum.} The idea is that indices stand for
tokens without associated value: where $(i,\alpha)$ can be read as ``the
element $\alpha$ in $A_i$'', $i$ represents some undetermined element of
which we only know it is in $A_i$. Then, for all $i\in I$, we set $\inj^{\fam A}_i=
\set{(\emptyMulSet,i)}\union\set{(\mulSet\alpha,(i,\alpha))\st\alpha\in A_i}$.
Moreover, if $\fam f$ is a relation from $\fam A$ to $\fam B$, 
we set $\Loplus\fam f=\natInst{\id}I\union\Oplus\fam f$ so that  $\Loplus$ is a
continuous $I$-ary functor from $\Rel^I$ to $\Rel$. 
For all $i\in I$ and $\ms \alpha=\fullMulSet\alpha n\in\oc A_i$, we write
$i\ms\alpha$ for $\mulSet{(i,\alpha_1),\dotsc,(i,\alpha_n)}\in\oc\Loplus\fam
A$. Then, for all family $\fam f$ 
of multirelations such that $f_i\in\RelCC(A_i,B)$ for all $i\in I$, we define
$\cotuple{\fam f}=\set{(\mulSet i+i\ms\alpha,\beta)\st 
i\in I\land (\ms\alpha,\beta)\in f_i}$
and obtain $\cotuple{\fam f}\CCcomp{\inj_i}=f_i$ for all $i\in I$.

Notice however that $\cotuple{\fam f}$ is not
characterized by this property, since we have already remarked that $\Loplus$ is
not a coproduct in $\RelCC$. For instance, 
$\set{(\mulSet{i,i}+i\ms\alpha,\beta)\st
(\ms\alpha,\beta)\in f_i}$ behaves similarly (we just added a copy
of the token $i\in I$).
This \definitive{case definition} construction can be internalized
as a multirelation, by setting $\natInst\lcase{\fam A,B}=\labs{\fam
f}{\cotuple{\fam f}}\subseteq\With_{i\in I}\pars{A_i\CCarrow
B}\CCarrow\Loplus\fam A\CCarrow B$. More explicitly:
\[\natInst\lcase{\fam A,B}=\set{\pars{
  \mulSet{(i,(\ms\alpha,\beta))},
  \mulSet{i}+i\ms\alpha,
  \beta}
  \st i\in I\land\ms\alpha\in\oc{A_i}\land\beta\in B}\]
and we obtain:
\begin{lemma}\label{lemma:case:inj}
  For all $j\in I$, 
  $\fam f:\With_{i\in I}\pars{A_i\CCarrow B},s:A_j
  \jug \appl{\appl\lcase{\fam f}}{\pars{\appl{\inj_j}s}}
  = \appl{f_j}s:B$.
\end{lemma}
This provides a lazy implementation of sum types. For instance, we have
\[\appl{\appl{\natInst\lcase{\fam A,B}}{\tuple{\labs x b_i}_{i\in I}}}{\pars{\appl{\inj_j}a}}
=b_j\]
for all $\fam b\in\powerset B^I$, $j\in I$ and $a\subseteq A_j$, even
when $a$ is undefined, \ie{} $a=\emptyset$.

For all $i\in I$, the restriction $\natInst{\rest_i}{\fam A}$ is a quasi-functional
lax natural transformation from $\Loplus$ to $\Pi_i$. The same holds 
for the index relation $\indx\union\natInst{\id}I$ from $\Loplus$ to $E_I$.
We thus have a transport situation, which moreover defines 
a functor $\Loplus$ from $\Fin^I$ to $\Fin$, because it admits a shape relation:
$\indx$ itself (see Lemma~\ref{lemma:transport:situation}). We obtain $\web{\Loplus \fam{\cA}}=\Loplus{\fam{\web\cA}}$ 
and 
$\fin{\Loplus\fam\cA}=\set{J\union \sum_{i\in K}a_i\st J\union K\finSubseteq I\land
\forall i\in K,\ a_i\in\fin{\cA_i}}$.
This defines a functor suitable to interpret sum types in $\FinCC$ (although not a coproduct)
because injections and the case definition operator are finitary: 
\begin{lemma}
  \label{lemma:case:inj:finitary}
  For all finiteness spaces $\fam\cA$ and $\cB$,
  $\inj_i\in\fin{\cA_i\CCarrow\Loplus\fam \cA}$
  and 
	\[\lcase\in\fin{
	\Loplus\fam\cA\CCarrow\With_{i\in I}\pars{\cA_i\CCarrow \cB}\CCarrow \cB}.\]
\end{lemma}
\begin{proof}
  This is a direct application of the definitions.
\end{proof}

We call \definitive{algebraic datatype} any functor built from projection functors,
$\fsEmpty$, $\With$ and $\Loplus$. The most basic example of composite datatype
is that of booleans, $\fsBool=\fsEmpty\loplus\fsEmpty$: assuming this lifted
sum is indexed by the two point set $\set{\webTrue,\webFalse}$, $\fsBool$ is
the only finiteness space with $\web{\fsBool}=\set{\webTrue,\webFalse}$. The
injections $\natInst{\inj_{\webTrue}}{(\emptyset,\emptyset)}=\set{(\emptyMulSet,\webTrue)}$ and
$\natInst{\inj_{\webFalse}}{(\emptyset,\emptyset)}=\set{(\emptyMulSet,\webFalse)}$
are constant multirelations: up to the isomorphism
$\emptyset\CCarrow\web{\fsBool}\cong\web{\fsBool}$,
we thus consider their respective images $\true=\set{\webTrue}\in\fin{\fsBool}$
and $\false=\set{\webFalse}\in\fin{\fsBool}$ as the constructors of $\fsBool$.
Similarly, the case definition
$\natInst{\lcase}{(\emptyset,\emptyset),A}=\set{\pars{
  \mulSet{\webVar},
  \mulSet{(\webVar,(\emptyMulSet,\alpha))},
  \alpha}
  \st \webVar\in\set{\webTrue,\webFalse}\land\alpha\in A}$
corresponds with the conditional
$\natInst{\bcase}{A}=
\set{(\mulSet\webTrue,\mulSet\alpha,\emptyMulSet,\alpha)\st\alpha\in A}
\union\set{(\mulSet\webFalse,\emptyMulSet,\mulSet\alpha,\alpha)\st\alpha\in A}$
up to the isomorphisms $\emptyset\CCarrow A\cong A$ and 
$A\with A\CCarrow A\cong A\CCarrow A\CCarrow A$, so that
$\natInst{\bcase}{\web\cA}\in\fin{\fsBool\CCarrow\cA\CCarrow\cA\CCarrow\cA}$
for all finiteness space $\cA$.
Of course, we obtain $\bcase\true=\labs x{\labs y{x}}$ and 
$\bcase\false=\labs x{\labs y{y}}$.

\subsection{Tree types as fixpoints}
\label{section:datatypes:fixpoints}

Algebraic datatypes are $\finExt$-increasing functors,
and both exactly continuous and directly continuous. Hence they
admit least fixpoints, which are obtained as exact suprema by Lemma
\ref{lemma:continuity:fixpoint}.
We investigate how this construction could provide an interpretation of
recursive algebraic datatypes. We may first consider a finiteness space of
trees:
\begin{counter}
  Let $\cA$ and $\cB$ be finiteness spaces and
  $\cT:\cX\mapsto \cA\oplus(\cX\tensor\cB\tensor\cX)$.
  We can describe the least fixpoint $\fix\cT$ as follows:
  \begin{itemize}
    \item $\web{\fix\cT}$ is the set of all finite binary trees, 
      with leaves labelled by elements of $\web\cA$
      and nodes labelled by elements of $\web\cB$;
    \item a set $t$ of trees is finitary in $\fix\cT$ 
      when the set of all the labels of nodes (resp. leaves) of trees in $t$
      is finitary in $\cB$ (resp. $\cA$) and moreover the height of trees in $t$ is bounded.
  \end{itemize}
  Moreover, $\fix\cT$ is functorial in variables $\cA$ and $\cB$ because, by the above
  description, it can be defined directly as a transport functor. It should not
  however be considered as the datatype of binary trees with nodes of type $\cB$ and
  leaves of type $\cA$. Indeed, since this type relies on
  $\oplus$ which does not define a sum, we would also fail to define a
  suitable relational interpretation of pattern matching for this type
  of trees. Notice that this is not related with a finiteness argument:
  the same would hold for the relational model (or the coherence model for that matter).
\end{counter}

In light of this example, of the discussion on sums and of previous work on the
semantics of system $T$ \citep{vaux:relT}, we are led to study the finiteness properties of the
datatypes of trees obtained as fixpoints of power series functors:
\begin{definition}\label{def:trees}
Let $I$ be a set of indices, $\fam \cA=\genFam\cA iI$ a family of finiteness
spaces and $\fam J=\genFam JiI$ a family of sets of indices. We write
$\natInst\fsLazy{I,\fam J}\fam\cA$ for the least fixpoint of the algebraic 
functor $\cX\mapsto \Loplus_{i\in I}\cA_i\with \cX^{\with J_i}$, 
where $\cX^{\with J}$ denotes $\With_{j\in J}\cX$.
\end{definition}
We will in general simply write $\fsLazy\fam\cA$ for $\natInst\fsLazy{I,\fam
J}\fam\cA$. Intuitively $\fsLazy\fam\cA$ is the datatype of trees, in which
nodes of sort $i\in I$ are of arity $J_i$ and bear labels in $\cA_i$.
More precisely, we will show that $\web{\fsLazy\fam\cA}$ is the set of paths
starting from the root in such trees. Recall that terms are in general
interpreted by subsets of the web of their type: the subsets interpreting terms
of type $\fsLazy\fam\cA$ will be the sets of paths in the corresponding trees.
We will moreover obtain that these interpretations are all finitary.

In order to describe $\web{\natInst{\fsLazy}{I,\fam J}\fam\cA}$, we introduce
the associated construction $\natInst{\setLazy}{I,\fam J}$ in $\Rel$:
$\natInst\setLazy{I,\fam J}\fam A$ is defined as the least fixpoint of the
continuous functor $T:X\mapsto\Loplus_{i\in I}A_i\with X^{\with J_i}$, 
\ie{} $\setLazy_{I,\fam J}\fam A=\Union_{n\in\nat} T^n\emptyset$, 
which we simply write $\setLazy\fam A$ in general. Since $T$ is actually 
an $(I+1)$-ary continuous functor, $\setLazy$ is itself an $I$-ary continuous
functor in $\Rel$. Before we inspect the general form of $\setLazy\fam A$, 
let us first give an intuitive account of the binary case:
\begin{example}\label{example:trees}
  Consider the finiteness space $\fsBT=\natInst\fsLazy{I,\fam J}$
  obtained by Definition \ref{def:trees},
  where we set $I=\set{\lf,\nd}$, $J_\lf=\emptyset$ and $J_\nd=\set{\lnode,\rnode}$:
  this is the least fixpoint of the functor
  $\cX\mapsto\cA\loplus\pars{\cB\with(\cX\with\cX)}$.
  This is meant to represent the datatype of binary trees with leaves of type
  $\cA$ and nodes of type $\cB$. 
  The various indices can be interpreted as follows: 
  $\lf$ denotes a leaf whereas $\nd$ denotes an internal node;
  $\lnode$ denotes the left child of a node,
  whereas $\rnode$ denotes its right child.
  Then the elements of $\web{\fsBT}$ are sequences  
  of the following four shapes:
  \begin{itemize}
    \item $(\nd,(2,(j_1,(\nd,(2,(j_2,\dotsc,(\nd,(2,(j_n,\nd)))\cdots))))))$ where
      $j_k\in\set{\lnode,\rnode}$ for all $k$, which denotes a path to an
      internal node;
    \item $(\nd,(2,(j_1,(\nd,(2,(j_2,\dotsc,(\nd,(2,(j_n,(\nd,(1,\beta)))))\cdots))))))$ where
      $\beta\in\web\cB$, which denotes a path to an internal node,
      with a value in the interpretation of this node;
    \item $(\nd,(2,(j_1,(\nd,(2,(j_2,\dotsc,(\nd,(2,(j_n,\lf)))\cdots))))))$,
      which denotes a path to a leaf;
    \item $(\nd,(2,(j_1,(\nd,(2,(j_2,\dotsc,(\nd,(2,(j_n,(\lf,\alpha))))\cdots))))))$ 
      where $\alpha\in\web\cA$,
      which denotes a path to a leaf, with a value
      in the interpretation of the label of the leaf.
  \end{itemize}
  It makes only sense to adopt a more compact notation and write, \eg, 
  $\nd j_1\nd j_2\cdots\nd j_n\nd$ for
	\[(\nd,(2,(j_1,(\nd,(2,(j_2,\dotsc,(\nd,(2,(j_n,\nd)))\cdots)))))),\] and 
  similarly $\nd j_1\nd j_2\cdots\nd j_n\nd\beta$ for
	\[(\nd,(2,(j_1,(\nd,(2,(j_2,\dotsc,(\nd,(2,(j_n,(\nd,(1,\beta)))))\cdots)))))).\]

  Then let $a,a',a''$ be constants of 
  type $\cA$ and $b,b'$ be constants of type $\cB$, \ie{} $a,a',a''\in\fin\cA$
  and $b,b'\in\fin\cB$. The tree:
  \begin{center}
   \synttree[b[a][b'[a'][a'']]]
  \end{center}
  will be interpreted by the subset:
  \begin{multline*}
  \set{\nd}\union\set{\nd\beta\st\beta\in b}
  \union
  \set{\nd\lnode\lf}
  \union
  \set{\nd\lnode\lf\alpha\st\alpha\in a}
  \union
  \set{\nd\rnode\nd}
  \union
  \set{\nd\rnode\nd\beta\st\beta\in b'}\\
  \union
  \set{\nd\rnode\nd\lnode\lf}
  \union
  \set{\nd\rnode\nd\lnode\lf\alpha\st\alpha\in a'}
  \union
  \set{\nd\rnode\nd\rnode\lf}
  \union
  \set{\nd\rnode\nd\rnode\lf\alpha\st\alpha\in a''}.
  \end{multline*}
\end{example}

Let us turn to the general case.
By its definition, $\setLazy\fam A$
is the least set such that: $I\subseteq\setLazy\fam A$;
$(i,(1,\alpha))\in\setLazy\fam A$ for all $i\in I$ and $\alpha\in A_i$;
and $\pars{i,(2,(j,\tau))}\in\setLazy\fam A$ for all
$i\in I$, $j\in J_i$ and $\tau$ in $\setLazy\fam A$. Hence the general form of
an element $\tau\in \setLazy\fam A$ is:
$\tau=(i_1,(2,(j_1,\dotsc(j_n,i_{n+1})\cdots)))))$ or 
$\tau=(i_1,(2,(j_1,\dotsc(j_n,(i_{n+1},(1,\alpha)))\cdots)))))$ where
$j_k\in J_{i_k}$ for all $k\le n$ and, in the second case, 
$\alpha\in A_{i_{n+1}}$.
As in the above example, we introduce the following conventions for the sole
purpose of making this description of the elements $\setLazy \fam A$ more
reasonable.
We call \definitive{addresses}
all finite sequences $i_1j_1i_2j_2\cdots i_nj_n$ such that 
$j_{k}\in J_{i_k}$ for all $k\le n$ and 
write $\addresses$ for the set of all addresses.
We call \definitive{value} any element $\nu$ of $\Loplus\fam A$.
We say $\nu$ is of type $i\in I$ if $\nu=i$ or $\nu=(i,\alpha)$ 
with $\alpha\in A_i$: we then write $\type\nu=i$.
A \definitive{path} is the data $\pi\nu$ of an
address and a value.
We may factor prefixes out of multisets of paths or addresses: for instance, 
if $\ms\tau=\fullMulSet\tau n$ is a multiset of paths, 
we may write $ij\ms\tau=\mulSet{ij\tau_1,\dotsc,ij\tau_n}$.
Then $\setLazy \fam A$ is in bijection with the set of
all paths: from now on we consider 
$\setLazy\fam A$, and thus $\fsLazy\fam\cA$, up to this bijection.

Notice that the relation
$\natInst{\val_i}{\fam A}=\set{(\pi i\alpha,\alpha)\st
\pi\in\addresses \land \alpha\in A_i}$ is a quasi-functional lax
natural transformation from $\setLazy$ to the projection functor $\Pi_i$, for all $i\in I$.
Moreover, the relation $\natInst{\len}{\fam A}=\set{(i_1j_1\cdots i_nj_n\nu,n)\st
n\in\nat\land i_1j_1\cdots i_nj_n\nu\in \setLazy\fam A}$
is a quasi-functional
lax natural transformation from $\setLazy$ to $E_\nat$ where $\nat$ is the functor of shapes 
defined by natural numbers (see Section~\ref{subsec:transport:functor}). 
\begin{example}
In the setting of Example~\ref{example:trees}, 
we obtain: $\relAppl{\val_\nd}t=b\union b'$, 
$\relAppl{\val_\lf}t=a\union a'\union a''$ and 
$\relAppl{\len}t=\set{0,1,2}$.
\end{example}

We have thus given a precise account of the web $\web{\fsLazy\fam\cA}=\setLazy\fam{\web\cA}$.
Moreover, since $\fsLazy\fam\cA$ is defined as the least fixpoint of the algebraic
functor $\cT$ given in Definition \ref{def:trees}, and this fixpoint is 
an exact supremum, we obtain: $\fin{\fsLazy\fam\cA}=\Union_{n\in\nat}\fin{\cT^n\fsZero}$.
We can thus characterize this finiteness structure as follows:
\begin{lemma}
  \label{lemma:finitary:trees}
  Let $t$ be a set of paths. Then $t\in\fin{\fsLazy\fam\cA}$ iff
  $\relAppl{\len}t\in\PFin{\nat}$ and $\relAppl{\val_i}t\in\fin{A_i}$ for all $i\in I$.
  Moreover, if $\relAppl\len t$ is finite and, for all $i\in I$, 
  $\relAppl{\val_i}$ is finite, then $t$ is itself finite. 
\end{lemma}
We could thus have presented $\fsLazy$ equivalently as the functor
of paths, with web functor $\setLazy$, finiteness structure being transported by
$\fam{\val}$ and $\len$ (see Lemma~\ref{lemma:transport:situation}). It is
important to notice that only the above careful explicitation of the structure
of $\fsLazy$ allowed us to deduce this functoriality. At the time of writing, it
is unclear to us whether this technique generalizes to a larger class of
transport functors.
\begin{example}
  Lemma \ref{lemma:finitary:trees} implies that the interpretation 
  of the binary tree of Example~\ref{example:trees} is finitary in $\fsBT$.
\end{example}

\subsection{The finitary datatype of trees}
\label{section:datatypes:trees}

We are now ready to describe the interpretation of the datatype of trees:
\begin{itemize}
  \item $\fsLazy$ provides a \emph{lazy} implementation of the datatype of trees
    where nodes of type $i$ bear labels in $A_i$ and have arity $J_i$;
  \item this implementation is \emph{finitary} in the sense that constructors,
    destructors and iterators on trees are finitary relations.
\end{itemize}

The \definitive{lazy tree constructor} $\node_i\subseteq A_i\CCarrow \pars{\setLazy\fam
A}^{\with J_i}\CCarrow \setLazy\fam A$ is given by:
\[\node_i=\set{(\emptyMulSet,\emptyMulSet,i)}
\union\set{(\mulSet\alpha,\emptyMulSet,i\alpha)\st\alpha\in\web{\cA_i}}
\union\set{(\emptyMulSet,\mulSet{(j,\tau)},ij\tau)\st j\in J_i\land\tau\in
\setLazy\fam A}\] which is actually an instance of \[\inj_i\subseteq
\pars{A_i\with \pars{\setLazy\fam A}^{\with J_i}}\CCarrow \setLazy\fam A\] up to our
notations of addresses and the cartesian adjunction in $\RelCC$.
Since $\inj_i$ is finitary, we moreover obtain:
\[\node_i\in\fin{\cA_i\CCarrow\pars{\fsLazy\fam\cA}^{\with J_i}\CCarrow 
\fsLazy\fam\cA}\] for all family $\fam\cA$ of finiteness spaces.

\begin{example}
  Recall the finiteness space $\fsBT$ of Example~\ref{example:trees}.
  Notice that $\node_\lf\subseteq \cA\to\emptyset\to\fsBT$ 
 and $\node_\nd\subseteq\cB\to(\fsBT\with\fsBT)\to\fsBT$:
 up to standard isomorphisms, we consider the binary tree constructors
 $\leaf=\labs[\cA]x{\pars{\node_\lf x\tuple{}}}\subseteq \cA\to\fsBT$ 
 and $\node=\labs[\cB]y{\labs[\fsBT]t{\labs[\fsBT]u{\pars{\node_\nd
 y\tuple{t,u}}}}}\subseteq\cB\to\fsBT\to\fsBT\to\fsBT$. 
  Then the tree $t$ of Example~\ref{example:trees} is obtained as
	\[t=\node\, b \,(\leaf\, a)\,(\node\, b'\,(\leaf\, a')\,(\leaf\, a''))\]
  and we can check that the interpretation given there agrees with this
  identity.
\end{example}

Similarly, the \definitive{pattern matching operator}
is given by:
\begin{eqnarray*}
	\match&=&\bigg\{\bigg(
\mulSet{(\ms\alpha,\mulSet{(j_1,\tau_1),\dotsc,(j_n,\tau_n)},\beta)},
\mulSet{i}+i\ms\alpha+\sum_{k=1}^n\mulSet{ij_k\tau_k},
\beta\bigg)\st \\
&&\qquad\qquad\qquad i\in I\land\beta\in B\land\ms\alpha\in\oc{A_i}\land\forall k,\ j_k\in
J_i\land\tau_k\in\setLazy\fam A\bigg\}
\\&\subseteq&\With_{i\in
I}\pars{A_i\CCarrow\pars{\setLazy\fam
A}^{\with J_i}\CCarrow B}\CCarrow\setLazy\fam A\CCarrow B
\end{eqnarray*}
which is an instance of $\lcase\subseteq 
\With_{i\in I}\pars{\pars{A_i\with \pars{\setLazy\fam A}^{\with J_i}}\CCarrow B}\CCarrow\setLazy\fam A\CCarrow B$ up to our
notations of addresses and the cartesian adjunction in $\RelCC$.
As such, it is finitary: for all finiteness spaces 
$\fam\cA$ and $\cB$, 
$\match\in\fin{\fsLazy\fam \cA\CCarrow\With_{i\in I}
\pars{\cA_i\CCarrow\pars{\fsLazy\fam \cA}^{\with J_i}\CCarrow
\cB}\CCarrow  \cB}$.

As an application of Lemma~\ref{lemma:case:inj}, we moreover obtain
that pattern matching is correct:
\[\fam f:\With_{i\in I}\pars{A_i\CCarrow\pars{\setLazy\fam
A}^{\with J_i}\CCarrow B}, a:A_i, \fam t: \pars{\setLazy\fam
A}^{\with J_i}\jug \match{\fam f}\pars{\node_i a \fam t}=f_i a \fam
t:B\]
for all $i\in I$.
Similarly to that of sums, this encoding of trees is lazy in the sense that,
for all $\fam b\in\powerset{B}^I$ and $i\in I$,
$\match{\tuple{\labs x{\labs y{b_i}}}_{i\in I}}\pars{\node_i \emptyset
\tuple{\emptyset}_{i\in I}}=b_i$.

We can then construct the \definitive{iterator on trees}:
\[\begin{array}{rcl}
\iter&=&\fp\ \labs{F}{\labs{t}{\parLabs{\fam f}{\match\tuple{\labs a{\parLabs{\fam
t}{f_i\,a\tuple{F\,t_j \fam f}_{j\in
J_i}}}}_{i\in I}t}}}\\
&\subseteq&\setLazy\fam A\CCarrow \With_{i\in I} \pars{A_i\CCarrow
B^{\with J_i}\CCarrow B}\CCarrow B
\end{array}\]
which automatically satisfies 
\[
\fam f:\With_{i\in I} \pars{A_i\CCarrow B^{\with J_i}\CCarrow B},
a:A_i, \fam t:\pars{\setLazy\fam A}^{\with J_i}
\jug\iter \pars{\node_i a \fam t}\fam f=f_i a\tuple{\iter\,t_j\fam f}_{j\in J_i}.\]
The following lemma makes the structure of $\iter$ explicit: 
\begin{lemma}
  \label{lemma:iter}
  Let $\iter_0=\emptyset$ and, for all $n\in\nat$, let 
  \[\begin{array}{rcr}
    \iter_{n+1}&=&\set{\pars{
    \mulSet{i}+i\ms\alpha+\sum_{k=1}^p ij_k\ms\tau_k,
    \mulSet{\pars{i,\ms\alpha,\sum_{k=1}^p\mulSet{(j_k,\beta_k)},\beta}}
    +\sum_{k=1}^p\ms\phi_k, 
    \beta}
    \st\right.\quad\\ &&\quad\left.i\in I\land p\in\nat\land \forall k,\
    j_k\in J_k\land (\ms\phi_k,\ms\tau_k,\beta_k)\in
    \iter_n}.
  \end{array}\]
  Then $\family{\iter_n}{n\in\nat}$ is increasing for inclusion and
  $\iter=\Union_{n\in\nat} \iter_n$.
  Moreover, if
  $\iota=(\ms\phi,\ms\tau,\beta)\in\iter$, then
  $\iota\in\iter_{\max(\relAppl{\len}{\support{\ms\tau}})+1}$.
\end{lemma}
\begin{proof}
  The equation $\iter=\Union_{n\in\nat} \iter_n$ is just an unfolding of the
  definitions: if we write $f=\labs{F}{\labs{t}{\parLabs{\fam
  f}{\match\tuple{\labs a{\parLabs{\fam t}{f_i\,a\tuple{F\,t_j\fam f}_{j\in
  J_i}}}}_{i\in I}\,t}}}$ then $\iter=\fp f=\Union_{n\in\nat}f^n\emptyset$ 
  and we just have to check that $f^n\emptyset=\iter_n$ by induction on $n$.
  The additional result is straightforwardly deduced from this explicitation.
\end{proof}

We now relate precisely the indices and values in the input paths of 
$\iter$ with those used in the associated instance of iterated functions.
First, if $\tau=i_1j_1\cdots i_nj_n\nu\in\setLazy\fam A$, we write 
$\indices(\tau)=\set{i_1,j_1,\dotsc,i_n,j_n,\type{\nu}}$ and 
$\values(\tau)=\Union_{i\in I}\relAppl{\val_i}{\tau}$ which are both finite.
Moreover, if
$\phi=(i,\ms\alpha,\sum_{k=1}^n\mulSet{(j_k,\beta_k)},\beta)\in
\With_{i\in I}\pars{A_i\CCarrow B^{\with J_i}\CCarrow B}$,
we set $\indices(\phi)=\set i\union\set{j_k\st 1\le k\le n}$
and $\values(\phi)=\support{\ms\alpha}$.
We extend these to multisets by taking the union of images as in
$\indices(\ms\tau)=\Union_{\tau\in\support{\ms\tau}}\indices(\tau)$.
Recall that if $\ms\alpha\in\oc A$, 
$\card\ms\alpha$ denotes the multiset cardinality of $\ms\alpha$.
When $\ms\phi=\sum_{k=1}^n\mulSet{(i_k,\ms\alpha_k,\ms\tau_k,\beta_k)}
\in\oc{\With_{i\in I}\pars{A_i\CCarrow B^{\with J_i}\CCarrow B}}$, 
we write $\subsize(\ms\phi)=\sum_{k=1}^n\card\ms\alpha_k$.
\begin{lemma}
  For all $\iota=(\ms\tau,\ms\phi,\beta)\in\iter$,
  we have:
  \begin{itemize}
    \item $\indices(\ms\tau)=\indices(\ms\phi)$;
    \item $\values(\ms\tau)=\values(\ms\phi)$;
    \item $\card\ms\tau=\card\ms\phi+\subsize\ms\phi$.
  \end{itemize}
\end{lemma}
\begin{proof}
  The result is easily established for all $\iota\in\iter_n$, 
  by induction on $n$.
\end{proof}

\begin{lemma}
  Iteration is finitary: 
  $\iter\in\fin{\fsLazy\fam\cA\CCarrow
  \With_{i\in I}\pars{\cA_i\CCarrow\cB^{\with J_i}\CCarrow \cB}
  \CCarrow \cB}$.
\end{lemma}
\begin{proof}
  If $t\in\fin{\fsLazy\fam\cA}$,
  then $\relAppl{\len}{t}$ is finite:
  we write $n=\max(\relAppl{\len}{t})$.
  By Lemma \ref{lemma:iter},
  $\relAppl{\iter}{\prom t}=\relAppl{\iter_{n+1}}{\prom t}\in\fin{
  \With_{i\in I}\pars{\cA_i\CCarrow\cB^{\with J_i}\CCarrow \cB}
  \CCarrow \cB}$ because $\iter_{n+1}$ is finitary.
  Now fix $(\ms\phi,\beta)\in\web{\With_{i\in I}\pars{\cA_i\CCarrow\cB^{\with
  J_i}\CCarrow \cB} \CCarrow \cB}$ and 
  $\ms t'=\relAppl{\relRev{\iter}}{(\ms\phi,\beta)}$:
  we prove that $\ms t'\finPolar\prom t$.
  By the previous lemma, for all $\ms\tau\in \ms t'$,
  $\values(\ms\tau)=\values(\ms\phi)$,
  $\indices(\ms\tau)=\indices(\ms\phi)$,
  $\card\ms\tau=\card\ms\phi+\subsize\ms\phi$.
  Paths in $\support{\ms t'}\inter t$ have addresses of length at most $n$ with
  indices taken in a fixed finite set; moreover they hold values taken 
  in a fixed finite set. We deduce $\support{t'}\inter t$ is finite.
  Moreover, multisets in $\ms{t'}\inter \prom{t}$ are of fixed size:
  hence $\ms{t'}\inter \prom{t}$  is finite.
\end{proof}

Summing up the results in section \ref{section:datatypes:fixpoints} and
the current section, we obtain:
\begin{theorem}
  \label{theo:iteration}
  For all choice of sets of indices $I$ and $\fam J$, $\fsLazy\fam\cA$
  is the finiteness space of paths whose finiteness structure is transported by
  $\fam\val$ and $\len$. Moreover, there are multirelations $\node_i$ and
  $\iter$ such that:
  \begin{itemize}
    \item $\node_i\in\fin{\cA_i\CCarrow\pars{\fsLazy\fam\cA}^{\with
      J_i}\CCarrow\fsLazy\fam\cA}$;
    \item $\iter\in\fin{\fsLazy\fam\cA\CCarrow
      \With_{i\in I}\pars{\cA_i\CCarrow\cB^{\with J_i}\CCarrow \cB}
      \CCarrow \cB}$;
    \item $\labs a{\labs {\fam t}{\parLabs {\fam f}{\iter \pars{\node_i\,a
      \fam t}\fam f}}}=\labs a{\labs {\fam t}{\parLabs {\fam f}{f_i\,
      a\tuple{\iter\,t_j\fam f}_{j\in J_i}}}}$.
  \end{itemize}
  Hence $\fsLazy\fam\cA$ is the datatype of trees whose nodes of sort $i\in I$
  are labelled with values in $\cA_i$ and of arity $J_i$.
\end{theorem}

As an example of application of this theorem, consider the functor 
$\cT:\cX\mapsto\fsEmpty\loplus \cX$, that is obtained for
$I=\set{\lzero,\lsucc}$, $J_\lzero=\emptyset$, $J_\lsucc$ any singleton set,
and $\cA_\lzero=\cA_\lsucc=\fsEmpty$. Then
$\fix\cT\cong\fsLazyNat$ where $\web\fsLazyNat=\nat\union\mt\nat$,
$\fin{\fsLazyNat}=\PFin{\web{\fsLazyNat}}$ and 
$\mt\nat=\set{\mt n\st n\in\nat}$ is just a disjoint copy of $\nat$:
$n\in\nat$ (resp. $\mt n\in\mt\nat$) corresponds with the only path $\tau=\pi\nu$ such that 
$\relAppl{\len}\tau=\set n$ and $\type\nu=\lzero$ (resp. $\type\nu=\lsucc$).
The finiteness space $\fsLazyNat$ is intuitively that of \definitive{lazy
natural numbers}: $n$ stands for ``exactly $n$'' whereas $\mt n$ stands for
``strictly more that $n$''. From $\inj_0$ and $\inj_1$, we derive 
$\zero=\set 0\in\fin{\fsLazyNat}$ and $\succ=\set{(\emptyMulSet,\mt 0)}\union 
\set{(\mulSet\nu,\nu\plusOne)\st\nu\in\web{\fsLazyNat}}\in\fin{\fsLazyNat\CCarrow\fsLazyNat}$
where $n\plusOne=n+1$ and ${\mt n}\plusOne=\mt{(n+1)}$. Up to some standard
isomorphisms, we derive a variant $\natiter$ of $\iter$ such that:
\begin{itemize}
  \item $\natiter\in\fin{\fsLazyNat\CCarrow(\cA\CCarrow
    \cA)\CCarrow\cA\CCarrow\cA}$; 
  \item $\natiter\,\zero=\labs f{\labs x x}$;
  \item $\labs n{\pars{\natiter\pars{\succ\,n}}}=\labs n{\labs f{\labs
    x{\pars{f\pars{\natiter\,nfx}}}}}$.
\end{itemize}
This provides a finitary relational semantics of Gödel's system $T$, 
which shows that $\FinCC$ can accomodate the standard notion of computational
iteration. This was the subject of a previous article by the second author
\citep{vaux:relT} which moreover shows that the same can be done for the
recursor variant of system $T$.

The same applies here, actually: we could very well reproduce the results of
this section, replacing $\iter$ with 
\[\rec=\fp\pars{\labs{F}{\labs{t}{\labs{\fam f}{\pars{\match\tuple{\labs
a{\labs{\fam t}{\pars{f_i\, a\fam t\tuple{F\,t_j\fam f}_{j\in J_i}}}}_{i\in I}\,t}}}}}}\]
which automatically satisfies 
\[\labs a{\labs{\fam t}{\parLabs{\fam f}{\rec\pars{\node_i\, a\fam t}\fam f}}}
=\labs a{\labs{\fam t}{\parLabs{\fam f}{f_i\, a\fam t\tuple{\rec\,t_j\fam f}_{j\in J_i}}}}.\]
We would then verify that $\rec\in\fin{\fsLazy\fam\cA\CCarrow \With_{i\in
I}\pars{\cA_i\CCarrow{\fsLazy\fam\cA}^{\with J_i}\CCarrow\cB^{\with
J_i}\CCarrow \cB} \CCarrow \cB}$ for all finiteness spaces $\fam\cA$ and
$\cB$.

\bibliographystyle{dcu}
\bibliography{bib}

\end{document}

%% file: prelude.tex
\usepackage{url}
\usepackage{enumerate}
\usepackage{amssymb}
\usepackage{stmaryrd}
\usepackage{synttree}
\branchheight{1cm}

\usepackage{a4wide}

\usepackage[mathletters]{ucs}
\usepackage[utf8x]{inputenc}
\usepackage[T1]{fontenc}
\usepackage[unicode,hyperindex,colorlinks=true,citecolor=red,linkcolor=blue]{hyperref}

\input{pdfaliases}

\usepackage[english]{babel}

\usepackage[authoryear]{natbib}

\ifx\proof\relax\else\input{proofs}\fi

\input{mathcommon}

\usepackage{sets}

\input{calcaps}

\input{frakcaps}

\input{abbr.en}

\input{linlog}

\input{lambda}

\input{typing}

\input{rel}

\input{fin}

\input{localdefs}

%% file: proofs.tex
\usepackage{QED}
\let\carre\square
\def\qedsymbol{\carre}

%% file: mathcommon.tex
\usepackage{mathtools}

\newcommand*{\pars}			[1]{\left(#1\right)}
\newcommand*{\brackets}	[1]{\left[#1\right]}
\newcommand*{\braces}		[1]{\left\{#1\right\}}
\newcommand*{\angles}	[1]{\left\langle#1\right\rangle}

\newcommand*{\funArrow}{\longrightarrow}

\newcommand*\funDec[4][\funArrow]{ #2 \mathrel : #3 \mathrel{#1} #4}

\newcommand*{\inter}{\cap}
\newcommand*{\union}{\cup}
\newcommand*{\Inter}{\bigcap}
\newcommand*{\Union}{\bigcup}

\newcommand{\fullEl}[2]{#1_1,\dotsc,#1_{#2}}

\newcommand{\family}[2]{\pars{#1}_{#2}}
\newcommand{\genFam}[3]{\family{#1_#2}{#2\in#3}}

\newcommand*{\binRest}[3]{\mathchoice
              {\setbox1\hbox{${\displaystyle #1}^{\scriptstyle #3}_{\scriptstyle #2}$}
              \binRestAux{#1}{#2}{#3}}
              {\setbox1\hbox{${\textstyle #1}^{\scriptstyle #3}_{\scriptstyle #2}$}
              \binRestAux{#1}{#2}{#3}}
              {\setbox1\hbox{${\scriptstyle #1}^{\scriptscriptstyle #3}_{\scriptscriptstyle #2}$}
              \binRestAux{#1}{#2}{#3}}
              {\setbox1\hbox{${\scriptscriptstyle #1}^{\scriptscriptstyle #3}_{\scriptscriptstyle #2}$}
              \binRestAux{#1}{#2}{#3}}}
\newcommand*{\binRestAux}[3]{{#1\,\smash{\vrule height 1.1\ht1 depth 1.2\dp1}}^{\,#3}_{\,#2}}
\newcommand*{\funRest}[2]{\binRest{#1}{#2}{}}

\newcommand{\eqdef}{\stackrel{\mathrm{def}}=}

\newcommand*\transpose[1]{\prescript{t}{}{#1}}

%% file: calcaps.tex
\newcommand\cA{\mathcal{A}}
\newcommand\cB{\mathcal{B}}
\newcommand\cC{\mathcal{C}}

\newcommand\cF{\mathcal{F}}

\newcommand\cI{\mathcal{I}}

\newcommand\cL{\mathcal{L}}

\newcommand\cN{\mathcal{N}}

\newcommand\cP{\mathcal{P}}

\newcommand\cR{\mathcal{R}}
\newcommand\cS{\mathcal{S}}
\newcommand\cT{\mathcal{T}}
\newcommand\cU{\mathcal{U}}

\newcommand\cX{\mathcal{X}}
\newcommand\cY{\mathcal{Y}}

%% file: frakcaps.tex
\newcommand\fA{\mathfrak{A}}
\newcommand\fB{\mathfrak{B}}
\newcommand\fC{\mathfrak{C}}

\newcommand\fF{\mathfrak{F}}

\newcommand\fM{\mathfrak{M}}

\newcommand\fP{\mathfrak{P}}

\newcommand\fT{\mathfrak{T}}

%% file: abbr.en.tex
\newcommand{\ie}{{i.e.}}
\newcommand{\eg}{e.g.}
\newcommand{\st}{;\ }

%% file: linlog.tex
\usepackage{cmll}
\newcommand{\impl}{\mathbin{\Rightarrow}}
\newcommand{\limpl}{\mathbin{\multimap}}
\newcommand{\tensor}{\otimes}
\newcommand{\Tensor}{\bigotimes}
\newcommand{\Parr}{\bigparr}
\newcommand{\Oplus}{\bigoplus}
\newcommand{\With}{\bigwith}

\DeclareMathAlphabet{\mathLLrm}{T1}{cmr}{m}{n}

\newcommand\unit  {\mathLLrm u}

\newcommand\dere  {\mathLLrm d}

\newcommand\LLOne   {\mathbf 1}

\newcommand\LLZero  {\mathbf 0}


%% file: lambda.tex
\newcommand{\abstraction}[3]{\mathord{\mathord{#1#2}\,\mathord{#3}}}
\newcommand{\application}[2]{\mathord{\mathord{#1}\,\mathord{#2}}}

\newcommand{\abs}[2]{\abstraction{\lambda}{#1}{#2}}
\newcommand{\appl}[2]{\application{\left( #1 \right)}{#2}}
\newcommand{\subst}[3]{%
	\application{#1}{%
		\left[ #2 \mathbin{{:}{=}} #3 \right]%
	}%
}

\newcommand{\parAppl}[2]{\appl{#1}{\pars{#2}}}

\newcommand{\labs}[2]{\abs{#1}{#2}}

\newcommand{\parLabs}[2]{\labs{#1}{\pars{#2}}}

\newcommand{\zero}{\mathbf 0}

\newcommand{\lift}[1]{\left\uparrow #1\right .}

\newcommand{\lpi}[1]{\application{\lpiSym}{#1}}
\newcommand{\rpi}[1]{\application{\rpiSym}{#1}}
\newcommand{\pair}[2]{\angles{#1,#2}}



%% file: rel.tex
\DeclareSymbolFont{stmry}{U}{stmry}{m}{n}
\DeclareMathDelimiter\semlbracket{\mathopen}{stmry}{"4A}{stmry}{"71}
\DeclareMathDelimiter\semrbracket{\mathclose}{stmry}{"4B}{stmry}{"79}

\makeatletter
\newcommand{\sem}{\@ifnextchar[{\@semtwo}{\@semone}}
\def\@semone#1{\left\semlbracket#1\right\semrbracket}
\def\@semtwo[#1]#2{\relSem{#2}{#1}}
\makeatother
\newcommand{\relSem}[2]{\left\llbracket#1\right\rrbracket_{#2}}

\newcommand{\Rel}{\mathsf{Rel}}

%% file: fin.tex
\newcommand{\finSym}{\mathfrak{F}}
\newcommand{\fin}[1]{\finSym\left( #1 \right)}
\newcommand{\web}[1]{\ensuremath{\left|#1\right|}}

%% file: localdefs.tex
\renewcommand{\appl}[2]{\application{#1}{#2}}

\newcommand{\fsem}[1]{#1^*}
\newcommand{\parFsem}[1]{\fsem{\pars{#1}}}

\renewcommand{\eqdef}{=}
\renewcommand{\lpi}{\CCproj_1}
\renewcommand{\rpi}{\CCproj_2}

\newcommand*{\relAppl}[2]{#1\cdot #2}
\newcommand*{\relComp}[2]{#1\circ#2}
\newcommand*{\relRev}[1]{\transpose{#1}}
\newcommand*{\relDiv}[2]{#1\setminus#2}

\newcommand{\finExt}{\preceq}
\newcommand{\finExtSim}{\precsim}
\newcommand{\finInc}{\sqsubseteq}
\newcommand{\finSup}{\sqcup}
\newcommand{\FinSup}{\bigsqcup}
\newcommand{\finInf}{\sqcap}
\newcommand{\FinInf}{\bigsqcap}

\usepackage{slashed}
\declareslashed{\mathrel}{\scriptscriptstyle\bot}{0}{.25}{\finInc}
\declareslashed{\mathbin}{\scriptscriptstyle\bot}{0}{0}{\finSup}
\declareslashed{\mathop}{\scriptscriptstyle\bot}{0}{0}{\FinSup}
\declareslashed{\mathbin}{\scriptscriptstyle\bot}{0}{0}{\finInf}
\declareslashed{\mathop}{\scriptscriptstyle\bot}{0}{0}{\FinInf}
\newcommand{\dualInc}{\slashed{\finInc}}

\usepackage{optcmd}
\newcommand*{\polarNoOpt}{\mathrel{\bot}}
\newcommand*{\polarOpt}[1][]{\mathrel{\bot_{#1}}}
\newOptCmd{polar}{\polarOpt}{\polarNoOpt}
\newcommand*{\dualNoOpt}[1]{#1^{\bot}}
\newcommand*{\dualOpt}[2][]{#2^{\bot_{#1}}}
\newOptCmd{dual}{\dualOpt}{\dualNoOpt}
\newcommand*{\bidualNoOpt}[1]{#1^{\bot\bot}}
\newcommand*{\bidualOpt}[2][]{#2^{\bot\bot_{#1}}}
\newOptCmd{bidual}{\bidualOpt}{\bidualNoOpt}
\newcommand*{\tridualNoOpt}[1]{#1^{\bot\bot\bot}}
\newcommand*{\tridualOpt}[2][]{#2^{\bot\bot\bot_{#1}}}
\newOptCmd{tridual}{\tridualOpt}{\tridualNoOpt}
\newcommand*{\parDualNoOpt}[1]{\dual{\pars{#1}}}
\newcommand*{\parDualOpt}[2][]{\dual[#1]{\pars{#2}}}
\newOptCmd{parDual}{\parDualOpt}{\parDualNoOpt}
\newcommand*{\parBidualNoOpt}[1]{\bidual{\pars{#1}}}
\newcommand*{\parBidualOpt}[2][]{\bidual[#1]{\pars{#2}}}
\newOptCmd{parBidual}{\parBidualOpt}{\parBidualNoOpt}
\newcommand*{\parTridualNoOpt}[1]{\tridual{\pars{#1}}}
\newcommand*{\parTridualOpt}[2][]{\tridual[#1]{\pars{#2}}}
\newOptCmd{parTridual}{\parTridualOpt}{\parTridualNoOpt}

\newcommand\finSubscript{\mathrm{f}}
\newcommand{\finSubseteq}{\subseteq_{\finSubscript}}
\newcommand{\PFinSym}{\fP_{\finSubscript}}
\newcommand{\PFin}[1]{\PFinSym\left( #1 \right)}
\renewcommand{\MFinSym}{\fM_{\finSubscript}}
\renewcommand{\MFin}[1]{\MFinSym\left( #1 \right)}

\newcommand{\transport}[1]{\fF_{#1}}

\renewcommand{\lift}[2]{#1_{#2}}

\newcommand{\natInst}[2]{#1^{#2}}

\newcommand*{\addresses}{{\mathbf A}}

\newcommand{\thinSubsets}{\fT}

\newcommand{\proj}{\mathsf{proj}}
\newcommand{\CCproj}{\pi}
\newcommand{\rest}{\mathsf{rest}}
\newcommand{\indx}{\mathsf{indx}}
\newcommand*{\own}{{\mathsf{own}}}
\newcommand*{\supp}{{\mathsf{supp}}}
\newcommand*{\val}{{\mathsf{val}}}
\newcommand*{\shp}{{\mathsf{shp}}}
\newcommand*{\id}{{\mathsf{id}}}
\renewcommand*{\dere}{{\mathsf{der}}}
\newcommand{\inj}{{\mathsf{inj}}}
\newcommand{\lcase}{{\mathsf{case}}}

\newcommand{\len}{\mathsf{len}}
\newcommand{\node}{\mathsf{node}}
\newcommand{\leaf}{\mathsf{leaf}}
\newcommand{\match}{\mathsf{match}}
\newcommand{\iter}{\mathsf{iter}}
\newcommand{\rec}{\mathsf{rec}}
\newcommand{\natiter}{\mathsf{natiter}}
\newcommand{\fp}{\mathsf{fix}}
\newcommand{\size}{\mathsf{card}}
\newcommand{\true}{\mathsf{true}}
\newcommand{\false}{\mathsf{false}}
\newcommand{\bcase}{\mathsf{if}}
\renewcommand{\zero}{\mathsf{zero}}
\renewcommand{\succ}{\mathsf{succ}}

\newcommand{\fsEmpty}	  {\top}
\newcommand{\fsZero}	  {\LLZero}
\newcommand{\fsSgl}	    {\LLOne}
\newcommand{\fsBot}	    {\bot}
\newcommand{\fsFlatNat}	{\cN}
\newcommand{\fsLazyNat}	{\cN_l}
\newcommand{\mt}[1]{#1^>}
\newcommand{\plusOne}{^+}
\newcommand{\fsBool} {\cB ool}
\newcommand{\webTrue} {\doubleChar t}
\newcommand{\webFalse}{\doubleChar f}
\newcommand{\webVar}{\doubleChar x}
\newcommand{\fsBT}	{\cB\cT}
\newcommand{\nd}   {{\mathsf{N}}}
\newcommand{\lf}   {{\mathsf{F}}}
\newcommand{\lnode}{{\mathsf{G}}}
\newcommand{\rnode}{{\mathsf{D}}}
\newcommand{\lsucc}{{\mathsf{S}}}
\newcommand{\lzero}{{\mathsf{O}}}

\newcommand*{\setLazy}{{L}}
\newcommand{\fsLazy}    {\cL}

\newcommand{\fix}[1]{\mu\,#1}

\newcommand*{\typoCat}[1]{\underline{\mathrm{#1}}}
\renewcommand{\Rel}   {\typoCat{Rel}}
\renewcommand{\Fin}   {\typoCat{Fin}}
\newcommand{\RelCC}   {\Rel^\oc}
\newcommand{\FinCC}   {\Fin^\oc}

\newcommand{\loplus}  {\mathbin{\stackrel\bullet\oplus}}
\newcommand{\Loplus}  {\mathop{\stackrel\bullet\Oplus}}
\newcommand{\ocplus}{\oplus^\oc}

\newcommand*{\tuple}[1]{\angles{#1}}
\newcommand*{\cotuple}[1]{\braces{#1}}

\newcommand*{\indices}{ind}
\newcommand*{\values}{val}

\newcommand\subsize{\#\#}

\newcommand{\finPolar}{\polar[\finSubscript]}
\newcommand{\LambdaFin}   {\Lambda_{\Fin}}
\newcommand{\LambdaRel}   {\Lambda_{\Rel}}
\newcommand{\CCcomp}  {\circ^\oc}

\newcommand{\CCarrow}{\Rightarrow}

\newcommand{\doDoubleChar}[1]{\hbox{\hbox to 2pt{$#1$\hss}$#1$}}
\newcommand{\doubleChar}[1]{%
\mathchoice
{\doDoubleChar{\displaystyle#1}}
{\doDoubleChar{\textstyle#1}}
{\doDoubleChar{\scriptstyle#1}}
{\doDoubleChar{\scriptscriptstyle#1}}}

\newcommand{\prom}[1] {{#1}^\oc}

\newcommand{\bipr}    {\with}

\makeatletter
\newcommand{\Rec}{\@ifnextchar[{\@Rec}{\@@Rec}}
\def\@@Rec{\cR}
\def\@Rec[#1]{\@@Rec_{#1}}
\newcommand{\It}{\@ifnextchar[{\@It}{\@@It}}
\def\@@It{\cI}
\def\@It[#1]{\@@It_{#1}}
\newcommand{\Step}{\@ifnextchar[{\@Step}{\@@Step}}
\def\@@Step{\cS tep}
\def\@Step[#1]{\@@Step_{#1}}
\newcommand{\ItStep}{\@ifnextchar[{\@ItStep}{\@@ItStep}}
\def\@@ItStep{\cI t\cS tep}
\def\@ItStep[#1]{\@@ItStep_{#1}}
\newcommand{\Case}{\@ifnextchar[{\@Case}{\@@Case}}
\def\@@Case{\cC ase}
\def\@Case[#1]{\@@Case_{#1}}
\newcommand{\RecType}{\@ifnextchar[{\@RecType}{\@@RecType}}
\def\@@RecType{\textsl{Rec}}
\def\@RecType[#1]{\@@RecType\brackets{#1}}
\newcommand{\ItType}{\@ifnextchar[{\@ItType}{\@@ItType}}
\def\@@ItType{\textsl{Iter}}
\def\@ItType[#1]{\@@ItType\brackets{#1}}
\newcommand{\Approx}{\@ifnextchar[{\@Approx}{\@@Approx}}
\def\@@Approx#1{\cR^{\pars{#1}}}
\def\@Approx[#1]#2{\@@Approx{#2}_{#1}}
\newcommand{\Fix}{\@ifnextchar[{\@Fix}{\@@Fix}}
\def\@@Fix{\cF ix}
\def\@Fix[#1]{\@@Fix_{#1}}
\newcommand{\ItApprox}{\@ifnextchar[{\@ItApprox}{\@@ItApprox}}
\def\@@ItApprox#1{\cI^{\pars{#1}}}
\def\@ItApprox[#1]#2{\@@ItApprox{#2}_{#1}}
\newcommand{\ItStage}{\@ifnextchar[{\@ItStage}{\@@ItStage}}
\def\@@ItStage#1{\cS tage^{\pars{#1}}}
\def\@ItStage[#1]#2{\@@ItStage{#2}_{#1}}
\makeatother

\EnableBpAbbreviations

\newcommand{\supportSym}{supp}
\renewcommand{\support}[1]{\supportSym\pars{#1}}

\renewcommand{\card}{\mathop{\#}}

\newcommand{\ms}[1]{\overline{#1}}
\newcommand{\ag}[1]{\widetilde{#1}}
\renewcommand*{\fam}[1]{\overrightarrow{#1}}
\newcommand*{\maf}[1]{\overleftarrow{#1}}
\newcommand{\constants}{\fC}
\newcommand{\atoms}{\fA}
\renewcommand{\unit}{\angles{}}

\newcommand{\converts}{\simeq}

\makeatletter
\renewcommand{\labs}{\@ifnextchar[{\@typedLabs}{\untypedLabs}}
\renewcommand{\parLabs}{\@ifnextchar[{\@parTypedLabs}{\parUntypedLabs}}
\def\@typedLabs[#1]{\typedLabs{#1}}
\def\@parTypedLabs[#1]{\parTypedLabs{#1}}
\makeatother
\newcommand{\typedLabs}[3]{\abstraction{\lambda}{#2^{#1}}{#3}}
\newcommand{\untypedLabs}[2]{\abstraction{\lambda}{#1}{#2}}
\newcommand{\parTypedLabs}[3]{\typedLabs{#1}{#2}{\pars{#3}}}
\newcommand{\parUntypedLabs}[2]{\untypedLabs{#1}{\pars{#2}}}

\newcommand*{\LBL}[1]{\RightLabel{\small(#1)}}
\newcommand*{\SLBL}[1]{\RightLabel{$\sem{\text{#1}}$}}

\newcommand*{\definitive}[1]{\emph{#1}}

\newcommand*{\type}[1]{type\pars{#1}}

%% file: theorems.tex
\newtheorem{lemme}{Lemme}[section]
\newtheorem{lemma}[lemme]{Lemma}

\newtheorem{example}[lemme]{Example}

\newtheorem{theorem}[lemme]{Theorem}

\newtheorem{corollary}[lemme]{Corollary}

\newtheorem{proposition}[lemme]{Proposition}
\newtheorem{definition}[lemme]{Definition}

\newtheorem{counter}[lemme]{Counter-example}